\documentclass[journal]{IEEEtran}
\usepackage{etex}
\usepackage{graphicx}
\usepackage[square, numbers, comma, sort&compress]{natbib}
\usepackage{amsmath,amsfonts,amssymb,amscd,amsthm,xspace}
\usepackage{ifpdf}
\usepackage{xcolor}
\usepackage{array}
\usepackage[scriptsize]{subfigure}
\graphicspath{{./Figures/}}
\usepackage{authblk}
\usepackage[kerning=true]{microtype} 
\usepackage{float}  
\usepackage{hyperref}
\hypersetup{
	colorlinks   = true,
	citecolor    = blue,
	linkcolor    = black
}
\usepackage{fancyhdr}
\usepackage{booktabs}
\usepackage{cleveref}
\usepackage{bbm}
\usepackage{dsfont}
\usepackage{lipsum}
\usepackage{stmaryrd}
\usepackage{dcolumn}
\newcommand{\li}{n} 
\newcommand{\lo}{m}	
\newcommand{\invt}{\mathbf{X}}            
\newcommand{\ouvt}{\mathbf{Y}}          
\newcommand{\wuvt}{\mathbf{W}}
\newcommand{\delep}{\mathit{d}} 
\newcommand{\ivt}{\mathbbm{X}}   
\newcommand{\ovt}{\mathbbm{Y}}	
\newcommand{\wvt}{\mathbbm{W}}	
\newcommand{\pvt}{P_{\invt^{\li}}}
\newcommand{\hpvt}{P_{\widehat{\invt}^{\li}}} 
\newcommand{\qvt}{Q_{\invt}^{\li}}
\newcommand{\pnvt}{p}  
\newcommand{\bpvt}{P_{\overline{\invt}^{\li}}}

\newcommand{\mif}{I}
\newcommand{\infd}{i}

\newcommand{\real}{\mathbbm{R}}
\newcommand{\psp}{\mathbbm{P}}

\newtheorem{definition}{Definition}
\theoremstyle{theorem}
\newtheorem{theorem}{Theorem}
\newtheorem{hypothesis}{Hypothesis}
\newtheorem{lemma}{Lemma}

\newtheorem{remark}{Remark}

\makeatletter
\newcommand{\xRightarrow}[2][]{\ext@arrow 0359\Rightarrowfill@{#1}{#2}}
\newcommand{\xLeftrightarrow}[2][]{\ext@arrow 0359\Leftrightarrowfill@{#1}{#2}}
\makeatother

\begin{document}
\newcommand{\argmin}{\operatornamewithlimits{argmin}}
\newcommand{\argmax}{\operatornamewithlimits{argmax}}

\title{Maximum Likelihood Upper Bounds on the Capacities of Discrete Information Stable Channels}
\author{Tongxin~Li,~\IEEEmembership{Student Member,~IEEE
	} \thanks{Li is with the Computing + Mathematical Sciences Department, California Institute
		of Technology, Pasadena, CA 91125 USA (e-mail: tongxin@caltech.edu).}}
\maketitle

\begin{abstract}
Motivated by a greedy approach for generating {\it{information stable}} processes, we prove a universal maximum likelihood (ML) upper bound on the capacities of discrete information stable channels, including the binary erasure channel (BEC), the binary symmetric channel (BSC) and the binary deletion channel (BDC). The bound is derived leveraging a system of equations obtained via the Karush-Kuhn-Tucker conditions. Intriguingly, for some memoryless channels, \emph{e.g., }the BEC and BSC, the resulting upper bounds are {\it{tight}} and equal to their capacities. For the BDC, the universal upper bound is related to a function counting the number of possible ways that a length-$\lo$ binary subsequence can be obtained by deleting $\li-\lo$ bits (with $\li-\lo$ close to $\li \delep$ and $\delep$ denotes the {\it{deletion probability}}) of a length-$\li$ binary sequence. To get explicit upper bounds from the universal upper bound, it requires to compute a maximization of the matching functions over a Hamming cube containing all length-$\li$ binary sequences. Calculating the maximization exactly is hard. Instead, we provide a combinatorial formula approximating it. Under certain assumptions, several approximations and an {\it{explicit}} upper bound for deletion probability $\delep\geq 1/2$ are derived.

\end{abstract}



\section{Introduction}
\label{sec:0}

\IEEEPARstart{T}{he} \emph{information stable channels} were introduced by Dubrushin in ~\cite{dobrushin1963general}. Under the \emph{information stability} condition, sufficiently, their capacities can be expressed as
\begin{align}
\label{eq:1.0}
C=\liminf_{\li\rightarrow\infty}\frac{1}{\li}\ \sup_{\invt} {\ \mif\left( \invt ; \ouvt(\invt)\right)}.
\end{align}
Essentially, a channel satisfies information stability is equivalent to having the capacity expression above~\cite{hu1964shannon}. Preceding works have considered a variety of more general frameworks, \emph{e.g.},  a formula for channel capacity~\cite{verdu1994general} based on the information-spectrum method; a general capacity expression for channels with feedback~\cite{tatikonda2009capacity}; general capacity formulas for classical-quantum channels~\cite{hayashi2003general}, to list just a few.

Despite the simplicity of the formula in~(\ref{eq:1.0}), for some channels with memory, explicitly computing the capacities directly using the general formula is often not trivial. A famous example is the binary deletion channel (BDC), which was introduced by Levenshtein in~\cite{levenshtein1966binary} more than fifty years ago to model synchronization errors. In his model, a transmitter sends an infinite stream of bits representing messages over a communication channel. Before reaching at a receiver, the bits are deleted independently and identically with some {\emph{deletion probability}} $\delep\in (0,1)$.  The receiver wishes to recover the original message based on the deleted bits, with an asymptotically zero (in the length of the stream) probability of error. The BDC satisfies the information stability~\cite{dobrushin1967shannon}. Thus, the channel capacity denoted by $C(d)$ can be expressed via the formula in (\ref{eq:1.0}). However, a precise characterization of $C(\delep)$ is still unknown; it does not seem even possible to accurately compute the capacity numerically replying on the existing methods, for instance, the Blahut-Arimoto Algorithm (BAA)~\cite{arimoto1972algorithm,blahut1972computation,fertonani2010novel}.

In this work, we consider discrete channels with finite alphabets and derive a general upper bound (called the maximum likelihood (ML) upper bound in Section~\ref{sec:3}) on the capacities of information stables channels by analyzing a system of equations derived from the general formula in~(\ref{eq:1.0}). We demonstrate that for channels without memory, \emph{e.g.}, the binary erasure channel (BEC) and the binary symmetric channel (BSC). The corresponding upper bounds are tight for the BEC and BSC and equal to their channel capacities. For channels with memory, as a case study, we apply the ML upper bound to derive (\textit{implicit} and \textit{explicit}) approximations for $C(d)$, under certain assumptions. 

\subsection{Background}
\label{sec:1.1}

A discrete channel with a finite alphabet can be regarded as a stochastic matrix from an input space of all infinite-length sequences to an output space containing all sequences that can be obtained via the channel law. Formally, we follow the approach in \cite{vembu1995source,han1993approximation} and define the transmitted and received bit-streams via infinite processes.
For each fixed \textit{block-length} $\li$, there is a sequence of elements $\invt^{\left(\li\right)}_1 \ldots \invt^{\left(\li\right)}_{\li}$ selected from a finite set $ \mathcal{X}$, and there is a probability distribution $\pvt$ over this sequence. Let $\ivt$ denote an input process in terms of finite-dimensional sequences such that
$\ivt:= \{\invt^{\li}=(\invt_1^{(\li)},\ldots,\invt_\li^{(\li)})\}_{\li\geq 1}$. Similarly, denote by  $\ovt:=\{\ouvt^{\lo}=(\ouvt_1^{(\lo)},\ldots,\ouvt_\lo^{(\lo)})\}_{\lo\geq 1}$ with each $\ouvt_i^{(\lo)}$ in a finite set $\mathcal{Y}$ the corresponding output process of finite-dimensional sequences induced by $\ivt$ via the channel law $\wvt:=\{\wuvt^{\li}(\cdot|\cdot):\mathcal{X}^\li\rightarrow \mathcal{Y}^{\lo}\}_{\li,\lo\geq 1}$. So that
\begin{align*}
\mathbbm{P}\left(\mathbf{Y}^\lo = \mathbf{y}^\lo | \mathbf{X}^\li = \mathbf{x}^\li\right) := \wuvt^{\li}\left(\mathbf{y}^\lo|\mathbf{x}^\li\right).
\end{align*}

Note that the block-length of received codewords $\lo$ is not necessarily equal to $\li$, the block-length of the transmitted codeword. Moreover, the output block-length $\lo$ is allowed to be \emph{flexible}, meaning that it can be regarded as a random variable with distribution specified by the channel law~\footnote{Flexible output length allows us to apply this general framework to the BDC later in Section~\ref{sec:4}.}. In the remaining part of this paper, we often omit the superscript $\lo$ in $ \mathbf{y}^\lo$ and $ \mathbf{Y}^\lo$, to avoid confusion. This indicates the length of the output sequence $\mathbf{y}$ is not fixed.





\subsection{Outline of the Paper}
The remaining content of the paper is organized as follows. In Section~\ref{sec:1} we give a simplified version of $C(\delep)$, which is derived from the capacity formula in~(\ref{eq:cnw}) for information stable channels. Based on it, we prove a general upper bound (the ML upper bound in Theorem~\ref{thm:0}) on information stable channels in Section~\ref{sec:3}. Section~\ref{sec:3.3} follows by verifying the tightness of the ML upper bound for the BEC and the BSC. Next, in Section~\ref{sec:4}, several approximations for the capacity of the BDC are reported.



\section{Preliminaries}
\label{sec:1}

\subsection{Notational Convention}
We use $\log(\cdot)$ to denote logarithms over base $2$, unless stated otherwise. Let $\mathcal{X}^n$ and $\overline{\mathcal{Y}}$ denote the set of all possible length-$\li$ sequences and the set of all induced output sequences (having flexible lengths). Let $N:=\left|\mathcal{X}^n\right|$ and $M:=\left|\overline{\mathcal{Y}}\right|$.
We use the lowercase letter $j$ to index the $j$-th length-$\li$ input sequence $\mathbf{x}^{\li}_j$, and the letter $j$ to index the length-$\lo$ output sequence $\mathbf{y}^{\lo}_i$ with $j=1,\ldots,N$ and $i=1,\ldots,M$ respectively.  To distinguish between random variables and their realizations, we denote the former by capital letters and the latter by lower case letters, at most of the places throughout this work\footnote{Except for the output length $\lo$, which is a random variable dependent on the channel law.}. 

\subsection{Capacity Proxies}
For a fixed dimension $\li$, we maximize the mutual information between $\invt^\li\in\mathcal{X}^\li$ and $\ouvt\in\overline{\mathcal{Y}}$ in a way similar to defining the ``information capacity" for discrete memoryless channels (DMCs) over the binary alphabet, to obtain the quantity:
\begin{align}
\label{eq:cnw}
C_\li(\wuvt^\li):= \frac{1}{\li}\ \sup_{\invt^{\li}} {\ \mif\left( \invt^{\li} ; \ouvt\left( \invt^{\li}\right)\right)}
\end{align}
where the supremum is taken over all $\invt^\li\in\mathcal{X}^\li$ with distributions in the set
\begin{align}
\label{eq:1.1}
\psp^{N}
:=&\Big\{\pvt\in\real^{N}:\pnvt_j\geq 0 \ \forall \ j=1,\dots,N ; \  \sum_{j=1}^{N} \pnvt_j=1\Big\}.
\end{align}
 
 \subsection{Information-stability}

It turns out the quantity $C_\li(\wuvt^\li)$ is asymptotically (in $\li$) the same as the operational capacity under the following condition on channels, which is called \emph{information stability}\footnote{The way of classifying the channels that have an operational meaning with the capacity expressions in ($1.1$) using a condition called \emph{information stability} was first introduced by Dobrushin and Guoding Hu~\cite{dobrushin1963general,hu1964shannon}. It was restated and studied in many equivalent forms. For instance, in~\cite{vembu1995source}, information stability was proved to be insufficient to classify whether a source-channel separation holds or not. In~\cite{chen1999optimistic}, the expressions for optimistic channel capacity and optimal source coding rate are given for the class of information stable channels and similarly ``information stable" sources respectively.}.

\begin{definition}[Information Stability for Channels~\cite{dobrushin1963general,hu1964shannon,verdu1994general}]
\label{def:1}
A channel $\wvt$ is said to be {\em information stable}, if there exists an input process $\ivt$ such that $C_\li(\wuvt^\li)<\infty$ for all sufficiently large $\li$ and
\begin{align*}
\limsup_{\li\rightarrow \infty}\ \mathbbm{P}\left\{ \left|  \frac{\infd_{\invt^\li, \wuvt^\li}{\left(\invt^\li;\ouvt\left(\invt^\li\right)\right)}}{\li C_\li\left(\wuvt^\li\right)}-1\right  |>\gamma \right\}=0&\\
 \qquad\forall \ \gamma>0&
\end{align*}
where $\infd_{\invt^\li ,\wuvt^\li}{\left({\mathbf{x}}^{\li} ;{\mathbf{y}}\right)}:=\log\frac{{\wuvt^\li}({\mathbf{y}}|{\mathbf{x}}^{\li})}{{P}_{\ouvt}\left({\mathbf{y}}\right)}$ denotes the \emph{information density} for all ${\mathbf{x}}^{\li} \in\mathcal{X}^\li$ and ${\mathbf{y}} \in\overline{\mathcal{Y}}$. In other words, the normalized information density $\frac{1}{\li}\infd_{\invt^\li ,\wuvt^\li}{\left({\mathbf{x}}^{\li} ;{\mathbf{y}}\right)}$ converges \emph{in probability} to $C_\li(\wuvt^\li)$.
\end{definition}

Intuitively, information stability characterizes the types of channels in a manner similar to that the asymptotic equipartition property (AEP) characterizes stochastic sources. In fact, information stability for a channel $\wvt$ implies the existence of a class of corresponding input processes $\ivt$ such that $\ivt$, on being input to $\wvt$, results in a near-optimal code. For the case of discrete memoryless channels (DMCs), the operational meaning of the single-letter quantity in Eq. (\ref{eq:lim}) below appears as a natural consequence of the law of large numbers. For general channels, by considering the asymptotic behavior of the information density taking on an optimal input process $\ivt$ (which maximizes the mutual information), information stability provides (with sufficient generality, for a broader class of channels) an analogue of the law of large numbers. Such an optimal input process $\ivt$ may be understood to be equivalent to a sequence of codes that are capacity-achieving asymptotically in the block-length $\li$. The key idea relies on classical achievability bounds (for instance, Feinstein's lemma~\cite{feinstein1954new}, Shannon's achievability bound~\cite{shannon1957certain}).

Dobrushin in \cite{dobrushin1967shannon} proved that BDCs are information stable as defined in Definition~\ref{def:1}. For an arbitrary fixed $\li$, maximizing ${\mif\left(\invt^{\li} ; \ouvt\left(\invt^{\li}\right)\right)}$ in Eq.~(\ref{eq:cnw}) gives an optimal input distribution $\pvt^*$. Through appropriate achievability results (~\cite{feinstein1954new,shannon1957certain}), it is possible to construct an $(\li,{M},\lambda)$-code whose error probability vanishes as $\li$ goes to infinity. In addition, the rate ${\log {M}}/{\li}$ approaches $C_\li(\wuvt^\li)<\infty$  for sufficiently large $\li$. Hence for information stable channels, the capacities exist and can be written as\footnote{Note that this limiting expression does not always hold for general channels. For instance, consider one example in~\cite{verdu1994general}: a binary channel with output codewords equal to the input codeword with probability ${1}/{2}$ and changed independently of the input codewords with probability ${1}/{2}$. The capacity of this channel is $0$ since the error probability is always strictly positive and hence not vanishes. However, the formula in (\ref{eq:lim}) gives ${1}/{2}$.}
\begin{align}
\label{eq:lim}
C &=\liminf_{\li\rightarrow\infty}{C_\li(\wuvt^\li)<\infty}. 
\end{align}




\subsection{System of Equations for Optimality}

Recall $N:=\left|\mathcal{X}^\li\right|$ and $M:=\left|\mathcal{Y}\right|^\lo$. 

Our approach focuses on bounding $C_\li(\delep)$. Expressing the mutual information in terms of the channel law, the capacity-proxy $C_\li(\wuvt^\li)$ defined in~(\ref{eq:cnw}) equals to
\begin{align}
\nonumber
&C_\li(\wuvt^\li)\\
\label{eq:2.1}
=&\frac{1}{\li}{\sup_{\pvt}}{\sum_{\mathbf{x},\mathbf{y}}{{\pnvt\left(\mathbf{x}\right)\wuvt^{\li}\left(\mathbf{y}|\mathbf{x}\right)\log{\frac{\wuvt^{\li}\left(\mathbf{y}|\mathbf{x}\right)}{\sum_{\mathbf{x}}{\pnvt\left(\mathbf{x}\right)\wuvt^{\li}\left(\mathbf{y}|\mathbf{x}\right)}}}}}}.
\end{align}

Here, the supremum is taken over all distributions $\{\pnvt\left(\mathbf{x}\right)\}_{\mathbf{x}\in\mathcal{X}^{\li}}$ in the set $\psp^{N}$ (defined in~(\ref{eq:1.1})) and the summation is taken over all length-$\li$ input sequences $\mathbf{x}_j\in\mathcal{X}^{\li}$ and all output sequences  $\mathbf{y}_i\in\cup_{m}\mathcal{Y}^{\lo}$. 


From an optimization perspective, the asymptotic behavior of $C_\li(\wuvt^\li)$ can be captured by establishing a sequence of capacity-achieving distributions $\{\pvt^*\}_{\li}$ maximizing the following quantity for each $\li\geq 1$:
\begin{align}
\label{eq:2.21}
\sum_{j=1}^{N}\sum_{i=1}^{M}{{\pnvt_j\wuvt^{\li}\left(\mathbf{y}_i|\mathbf{x}_j\right)\log{\frac{\wuvt^{\li}\left(\mathbf{y}_i|\mathbf{x}_j\right)}{\sum_{j=1}^{N}{\pnvt_j\wuvt^{\li}\left(\mathbf{y}_i|\mathbf{x}_j\right)}}}}}.
\end{align}

Derived from the Karush-Kuhn-Tucker conditions, the following lemma generalizes Theorem 4.5.1. in~\cite{gallager1968information} (\emph{cf}.~\cite{jelinek1968probabilistic}), which was established to find channel capacities of DMCs with non-binary input/output alphabets. The lemma states a necessary and cufficient condition of the existence of $\{\pvt^*\}_{\li}$ maximizing (\ref{eq:2.21}) and it can be proved along the same line as in~\cite{gallager1968information}. The only difference is that for general channels, the summation is taken over all sequences in $\mathcal{Y}^\lo$ (this is in general exponential in $\lo$). While for DMCs, the summation can be decomposed and taken over the alphabet set of each individual coordinate of the sequence, thus the number of summations is linear in $\lo$.  For brevity the proof is omitted.  


\begin{lemma}[\cite{blahut1972computation,arimoto1972algorithm,gallager1968information,jelinek1968probabilistic}]
\label{thm:eq}
Fix a block-length $\li\geq 1$. There exists an optimal probability vector $\pvt^{*}=(\pnvt_1^*,\pnvt_2^*,\ldots,\pnvt_{N}^*)$ such that the quantity in (\ref{eq:2.21}) is maximized 
if and only if there exists $\lambda_{\li}\geq 0$ and for all $j=1,\ldots,N$,
\begin{align}
\nonumber
\frac{1}{\li}\sum_{i=1}^{M} \wuvt^{\li}\left(\mathbf{y}_i|\mathbf{x}_j\right)&\log{\frac{\wuvt^{\li}\left(\mathbf{y}_i|\mathbf{x}_j\right)}{\sum_{j=1}^{N}{\pnvt^*_j\wuvt^{\li}\left(\mathbf{y}_i|\mathbf{x}_j\right)}}}
\\
\label{eq:2.2}
&\begin{cases} =\lambda_{\li}&\mbox{   if } \pnvt_j^*\neq 0\\ \leq \lambda_{\li} &\mbox{      if } \pnvt_j^*=0\end{cases}.
\end{align}

Moreover the capacity
$C =\lim_{\li\rightarrow\infty}\lambda_{\li}$ if the limit exists.
\end{lemma}

Indeed,  (see~\cite{gallager1968information,jelinek1968probabilistic}) a probability distribution for an information stable source $\ivt$ satisfying~(\ref{eq:2.2}) always exists as $\li$ grows. Thus, the capacity-achieving distribution with fixed block-length $\li$ can be attained by solving the system~(\ref{eq:2.2}). Finding such an optimal $\pvt^*$ for the system of equations~(\ref{eq:2.2}) is equivalent to solving a non-linear system of equations that consists of exponentially (in $\li$) many variables. As introduced in Section~\ref{sec:0}, the BAA is one of the algorithms that can be applied to search for numerical solutions of~(\ref{eq:2.2}).


However, this approach has several limitations. On the one hand, in direct implementation of the BAA, as $\li$ grows, it becomes computationally intractable even to store the variables to be computed. One the other hand, as the BAA is itself an iterative algorithm attempting to solve the non-convex optimization problem~(\ref{eq:2.2}), and to the best of our knowledge for general channels, there are no guarantees on how quickly the numerical solution converges as a function of the number of iterations. Therefore, instead of looking for numerical answers, we concentrate on finding a universal upper bound on the capacities of general channels. This motivates the next section.



\section{Maximum Likelihood Upper Bound}
\label{sec:3}

In the sequel, we present some definitions. First, motivated by the notion of information stability defined in Definition~\ref{def:1}, we characterize a subset of the joint set $\mathcal{X}^\li \times\overline{\mathcal{Y}}$ consisting of all possible combinations of input and output sequences. This subset satisfies two vital properties. First, it behaves as a ``typical set'' and contains nearly all pairs of $\left(\mathbf{x}^\li,\mathbf{y}\right)$ randomly generated according to an arbitrary distributions $\pvt$ for every large $\li$. Second, conditioned on the pair $\left(\mathbf{X}^\li,\mathbf{Y}\right)$ belongs to the subset, the conditional mutual information does not differ too much from $C_\li\left(\wuvt^\li\right)$. Note that the concentration of information densities is stronger than that for information stable sources in two perspectives -- the concentration is in expectation; and it is required  to hold for every source $\ivt$.

\begin{definition}
	\label{def:6}
	For information stable channels with any source $\ivt$, a subset ${\mathcal{A}}$ of $\mathcal{X}^\li \times\overline{\mathcal{Y}}$ is called a \emph{concentration set} if it satisfies
	\begin{align}
	\label{eq:3.22}
	&\liminf_{\li\rightarrow\infty}\mathbbm{P}\left(\left(\mathbf{X}^\li,\mathbf{Y}\right)\in\mathcal{A}\right)=1,\\
	\label{eq:3.15}
	&\limsup_{\li\rightarrow\infty}\mathbbm{E}\left[\left|\frac{\infd_{\invt^\li, \wuvt^\li}{(\invt^\li;\ouvt\left(\invt^\li\right))}}{\li C_\li\left(\wuvt^\li\right)}-1\right|\Big | \left(\mathbf{X}^\li,\mathbf{Y}\right)\in\mathcal{A}\right]=0
	\end{align}
	where the randomness is over the source $\ivt$ and channel law $\wvt$.
\end{definition}

Later in Section~\ref{sec:3.3} and Section~\ref{sec:4.2.1}, we provide concrete and nontrivial examples of the concentration sets for the BEC, BSC and BDC respectively.

Since ${\mathcal{A}}\subseteq \mathcal{X}^\li \times\overline{\mathcal{Y}}$, the next lemma is straightforward.
\begin{lemma}
\label{lemma:2}
For each block-length $\li$, there exists a subset $\mathcal{B}$ of $\overline{\mathcal{Y}}$ such that
\begin{align*}
\mathcal{A}\subseteq \mathcal{X}^\li\times \mathcal{B}.
\end{align*}
\end{lemma}

It is useful to introduce the following ``constant'' version of the stochastic matrix $\wuvt$, called the {{\it stochastic factors}} for convenience. Again, we will carefully construct them in Section~\ref{sec:3.3} for both the BEC and the BSC, and in Section~\ref{sec:4.2} for the BDC.

\begin{definition}
\label{def:7}
We call a set of functions $f_k(\cdot|\cdot): \mathcal{Y}^{\lo}\times\mathcal{X}^{\li} \mapsto [0,1]$ {{\it stochastic factors}} if there exists a decomposition of $\mathcal{B}=\bigcup_{k\in\mathcal{K}}\mathcal{B}_k$ ($\mathcal{K}$ is a discrete set) such that
\begin{align}
\label{eq:3.11}
&\sum_{\mathbf{y}\in{\mathcal{B}_k}}f_k\left(\mathbf{y}|\mathbf{x}\right) = 1, \quad \forall \ \mathbf{x}\in\mathcal{X}^\li,  k\in\mathcal{K},\\
\label{eq:3.12}
&\sum_{k\in\mathcal{K}}\max_{\left(\mathbf{x},\mathbf{y}\right)\in\mathcal{A}_k}\frac{\wuvt^{\li}\left(\mathbf{y}|\mathbf{x}\right)}{f_k\left(\mathbf{y}|\mathbf{x}\right)}\leq 1
\end{align}
where $\mathcal{A}_k:=\mathcal{X}^\li\times\mathcal{B}_k$.
\end{definition}

Based on the concentration set and the stochastic factor defined above, we obtain the following upper bound on the capacity of an information stable channel:

\begin{theorem}[Maximum Likelihood Upper Bound\footnote{For intuition on why we call a maximum likelihood (ML) upper bound, see Section~\ref{sec:3.1.1}.}]
	\label{thm:0}
For a discrete information stable channel defined in Section~\ref{sec:1.1}, assume there exist a concentration set $\mathcal{A}$ and stochastic factors $f$ defined above. The following upper bound on the channel capacity  holds for any ${\mathcal{A}}=\bigcup_{k\in\mathcal{K}}\mathcal{X}^\li\times\mathcal{B}_k$ and $\{f_k\}_{k\in\mathcal{K}}$:
\begin{align}
\label{eq:2.6}
C\leq \liminf_{\li\rightarrow \infty}  \overline{C}_\li\left(\wuvt^\li\right)
\end{align}
where $\overline{C}_\li(\wuvt^\li)$ denotes the following quantity:
\begin{align}
\label{eq:3.6}
\overline{C}_\li(\wuvt^\li):={\frac{1}{\li}}\max_{k\in\mathcal{K}}\log\left( \sum_{\mathbf{y}\in {\mathcal{B}_k}}\max_{\mathbf{x}\in\mathcal{X}^{\li}}f_k\left(\mathbf{y}|\mathbf{x}\right)\right).
\end{align}
\end{theorem}

An intuitive derivation of the bound (\ref{eq:2.6}) is described below by formulating a simplified system in a greedy approach. The formal proof using Jensen's inequality is provided in Section~\ref{subsec:3.2}.

\subsection{Intuition}
\label{sec:3.1}
Recall that the system of equations in (\ref{eq:2.2}) gives, for every fixed dimension $\li$, an optimizing probability distribution $\pvt^*$ for the capacity proxy $C_\li(\wuvt^\li)$ in (\ref{eq:2.1}). Since actually solving the system of equations in (\ref{eq:2.2}) is computationally intractable for large $\li$, it is desirable to relax this system to a computationally tractable system that nonetheless provides a good outer bound to (\ref{eq:2.2}). Our starting point is the observation that an \emph{information stable} input process $\widehat{\ivt}=\{\widehat{\invt}^{\li}=(\widehat{\invt}_1^{(\li)},\ldots,\widehat{\invt}_\li^{(\li)})\}_{\li\geq 1}$ (with corresponding sequence of probability distributions $\{\hpvt=\widehat{p}_1,\ldots.\widehat{p}_{N}\}_{\li\geq 1}$) satisfies all but an asymptotically (in $\li$) vanishing fraction of the constraints in (\ref{eq:2.2}). To see this, one may notice that by the definition of information stability (in Definition~\ref{def:1}), for any fixed $\gamma>0$ it holds that there is a sufficiently large $N_\gamma$ such that for all $\li>N_\gamma$, in probability (over the input process $\ivt$ and the channel law $\wvt$ it holds that the ratio of the information density $\log\frac{{\wuvt^\li}({\mathbf{y}}|{\mathbf{x}}^{\li})}{{P}_{\ouvt}({\mathbf{y}})}$ over $\li C_\li(\wuvt^\li)$ converges to $1$. Moreover, for an arbitrary but fixed integer $\li$ and for all $\mathbf{x}\in\mathcal{X}^\li$,
\begin{align}
\label{eq:3.10}
\sum_{\mathbf{y}\in \overline{\mathcal{Y}}}\wuvt^{\li}\left(\mathbf{y}|\mathbf{x}\right)=1.
\end{align}

Therefore, for all sufficiently large $\li$, w.h.p. using the distribution of the information stable process $\widehat{\ivt}$, the quantity on the LHS of (\ref{eq:2.2}) is approximately equal to $C_\li(\wuvt^\li)$. Thus any information stable input process $\widehat{\ivt}$, for sufficiently large $\li$, becomes a reasonable approximation of the input process ${\ivt}^*$ optimizing (\ref{eq:2.2}).
This encourages us to construct a new input process $\overline{\ivt}$ by maximizing, for every integer $\li$, the probability in (\ref{eq:3.1}) below (using a greedy approach):
\begin{align}
\label{eq:3.1}
\mathbbm{P}\left\{\left|  \frac{\infd_{\overline{\invt}^\li, \wuvt^\li}{(\overline{\invt}^\li;{\ouvt}\left(\invt^\li\right))}}{\li C_\li(\delep)}\right  |=1\right\}.
\end{align}

Through this process we are able to introduce such a process $\overline{\ivt}$ (with corresponding distribution $\bpvt$) that mimics the one for information stability in Definition~\ref{def:1}.

\subsubsection{Approximate Information Stable Processes}
\label{sec:3.1.1}
To find a system that obtains such a sub-optimal input distribution $\bpvt$ efficiently, one simple heuristic method is to maximize the probability in (\ref{eq:3.1}) greedily. 


 For fixed input block-length $\li$, we consider the set of all output sequences $\mathbf{y}$ in the concentration set ${\mathcal{B}}$. For each $\mathbf{y}$ in ${\mathcal{B}}$, we \emph{greedily} choose the corresponding $\mathbf{x} \in \mathcal{X}^\li$ that maximizes the {\it a posteriori} probability of an instance $\mathbf{x}$ being transmitted under the channel law $\wvt$ (this is intuitively where the term $\max_{\mathbf{x}}f_k\left(\mathbf{y}|\mathbf{x}\right)$ comes from) that shows up in Eq.~(\ref{eq:3.6}).\footnote{This procedure coincides with the greedy decoding suggested in~\cite{mitzenmacher2009survey} for the BDC, which can be used to derive lower bounds on $C(\delep)$. However, making use of the system~(\ref{eq:2.2}), the greedy selection is also capable to give upper bounds.}

Now, guided by the intuition in the previous paragraph about the LHS of (\ref{eq:2.2}) being approximately equal to $C_\li(\wuvt^\li)$ for many $j$, we {\it fix} the information density 
\begin{align*}
\infd_{\invt^\li ,\wuvt^\li}{\left({\mathbf{x}}_j ;{\mathbf{y}}_i\right)}=\log{\frac{\max_{\mathbf{x}\in\mathcal{X}^\li}f_k\left(\mathbf{y}_i|\mathbf{x}\right)}{\sum_{j=1}^{\li}{\overline{\pnvt}^{\li}_j\wuvt^{\li}\left(\mathbf{y}_i|\mathbf{x}_j\right)}}} 
\end{align*}
for each such $({\mathbf{x}}_j,{\mathbf{y}}_i)$
drawn in above to equal a certain constant $\overline{\lambda}_{\li}$ (which shows up later, in Eq.~(\ref{3.2})).\footnote{We do not claim to have an efficient computational process for determining this constant $\lambda(\li)$. However, this $\overline{\lambda}_{\li}(\delep)$ has a strong operational meaning -- it provides an outer bound on the capacity $C(\delep)$ of the deletion channel, as discussed in Eq.~(\ref{eq:3.2}).} The fixing of the information density $\infd_{\invt^\li ,\wuvt^\li}{({\mathbf{x}}_j ;{\mathbf{y}}_i)}$ is done in a manner such that, using Bayes' rule, a probability distribution ${P}_{\overline{\invt}^{\li}}=(\overline{\pnvt}^{\li}_j)_{j=1,\ldots,N}$ is induced on $\mathbf{x}_j$. In particular, the value of $\overline{\lambda}_{\li}$ is chosen so that the summation of $\overline{\pnvt}^{\li}_j$ over all $\mathbf{x}_j$ equals $1$.

\subsubsection{Simplified System}

Formally, we describe the new (as a  simplified version of~(\ref{eq:2.2})) system as follows.

For all $\mathbf{y}\in{\mathcal{B}}_k$, $k\in\mathcal{K}$ and some $\overline{\lambda}_{\li}\geq 1$, we let  
\begin{align}
\label{3.2}
\frac{1}{\li}\log{\frac{\max_{\mathbf{x}\in\mathcal{X}^\li}f_k\left(\mathbf{y}|\mathbf{x}\right)}{\sum_{j=1}^{\li}{\overline{\pnvt}^{\li}_j\wuvt^{\li}\left(\mathbf{y}|\mathbf{x}_j\right)}}} &=\overline{\lambda}_{\li}.
\end{align}

As explained above, by exhausting the set ${\mathcal{B}}_k$, the constraints in (\ref{3.2}) suggest a \emph{greedy} approach for finding the sub-optimal distribution $\bpvt$. 

%

Recall that (Definition~\ref{def:7} in Section~\ref{sec:3}) $\sum_{\mathbf{y}\in{\mathcal{B}_k}}f_k\left(\mathbf{y}|\mathbf{x}\right) = 1$ for all $\mathbf{x}\in\mathcal{X}^\li$ and $k\in\mathcal{K}$. Given an input process $\overline{\ivt}$ satisfying Eqs.~(\ref{3.2}) for each integer $\li$, we can rewrite the constraints in (\ref{3.2}) as
\begin{align*}
\frac{\max_{\mathbf{x}\in\mathcal{X}^\li}f_k\left(\mathbf{y}|\mathbf{x}\right)}{\sum_{j=1}^{\li}{\overline{\pnvt}^{\li}_j\wuvt^{\li}\left(\mathbf{y}|\mathbf{x}_j\right)}}=2^{\li\overline{\lambda}_{\li}}, \quad \forall \ \mathbf{y}\in\mathcal{A}_k.
\end{align*}

Summing both sides over all $\mathbf{y}\in\mathcal{B}_k$,
\begin{align*}
\sum_{\mathbf{y}\in\mathcal{B}_k}\frac{\max_{\mathbf{x}\in\mathcal{X}^\li}f_k\left(\mathbf{y}|\mathbf{x}\right)}{2^{\li\overline{\lambda}_{\li}}}&=\sum_{\mathbf{y}\in\mathcal{B}_k}\sum_{j=1}^{N}{\overline{\pnvt}^{\li}_j\wuvt^{\li}\left(\mathbf{y}|\mathbf{x}_j\right)}\\
&=\sum_{j=1}^{N}\overline{\pnvt}^{\li}_j\sum_{\mathbf{y}\in\mathcal{B}_k}\wuvt^{\li}\left(\mathbf{y}_i|\mathbf{x}_j\right)\\
&=\sum_{j=1}^{N}\overline{\pnvt}^{\li}_j=1.
\end{align*}

Multiplying both sides with $2^{\li\overline{\lambda}_{\li}}$ and taking logarithms,
\begin{align}
\label{eq:3.2}
\overline{C}_\li(\wuvt^\li):={\frac{1}{\li}}\max_{k\in\mathcal{K}}\log\left( \sum_{\mathbf{y}\in {\mathcal{B}_k}}\max_{\mathbf{x}\in\mathcal{X}^{\li}}f_k\left(\mathbf{y}|\mathbf{x}\right)\right)=\overline{\lambda}_{\li}.
\end{align}

The number of constraints in (\ref{3.2}) is much smaller than in (\ref{eq:2.2})), suggesting $\overline{\lambda}_{\li}\geq \lambda_{\li}$
(without a proof) for each $\li$. This indicates that the ML upper bound defined in Theorem~\ref{thm:1} makes sense. Next we prove Theorem~\ref{thm:0}.

\subsection{Proof of Theorem~\ref{thm:0}}
\label{subsec:3.2}

\begin{figure*}[t]
	\begin{align}
	\label{eq:3.4}
	C_\li(\wuvt^\li)\leq{\frac{1}{\li}\sum_{\left(\mathbf{x}_j,\mathbf{y}\right)\in {\mathcal{A}}}q_j\wuvt^{\li}\left(\mathbf{y}|\mathbf{x}_j\right) {\log{\frac{\wuvt^{\li}\left(\mathbf{y}|\mathbf{x}_j\right)}{\sum_{j=1}^{N}{\pnvt_j^*\wuvt^{\li}\left(\mathbf{y}|\mathbf{x}_j\right)}}}}}+\gamma_n.
	\end{align}
	\hrule
\end{figure*}

\begin{figure*}[t]
\begin{align}
\label{eq:3.13}
&\sum_{\left(\mathbf{x}_j,\mathbf{y}\right)\in {\mathcal{A}}}q_j\wuvt^{\li}\left(\mathbf{y}|\mathbf{x}_j\right) {\log{\frac{\wuvt^{\li}\left(\mathbf{y}|\mathbf{x}_j\right)}{\sum_{j=1}^{N}{\pnvt_j^*\wuvt^{\li}\left(\mathbf{y}|\mathbf{x}_j\right)}}}}\leq \sum_{k\in\mathcal{K}}\sum_{\left(\mathbf{x}_j,\mathbf{y}\right)\in {\mathcal{A}_k}}\frac{\wuvt^{\li}\left(\mathbf{y}|\mathbf{x}_j\right)}{f_k\left(\mathbf{y}|\mathbf{x}_j\right)}q_j f_k\left(\mathbf{y}|\mathbf{x}_j\right) {\log{\frac{\wuvt^{\li}\left(\mathbf{y}|\mathbf{x}_j\right)}{\sum_{j=1}^{N}{\pnvt_j^*\wuvt^{\li}\left(\mathbf{y}|\mathbf{x}_j\right)}}}}.
\end{align}
	\hrule
\end{figure*}

Denote by $\pvt^*$ the optimizing probability distribution maximizing the quantity in (\ref{eq:2.21}). Based on Definition~\ref{def:6}, Lemma~\ref{lemma:2} and Definition~\ref{def:7}, we prove Theorem~\ref{thm:0}.

Considering the constraints in (\ref{eq:2.2}),  it follows that
\begin{align}
\label{eq:3.3}
&C_\li(\wuvt^\li)=\frac{1}{\li}\sum_{i=1}^{M} \wuvt^{\li}\left(\mathbf{y}_i|\mathbf{x}_j\right)\log{\frac{\wuvt^{\li}\left(\mathbf{y}_i|\mathbf{x}_j\right)}{\sum_{j=1}^{N}{\pnvt_j^*\wuvt^{\li}\left(\mathbf{y}_i|\mathbf{x}_j\right)}}}
\end{align}
for all $ j\in\{1,\ldots,N\}$ with $\pnvt_j^*\neq 0$.


Now we introduce an auxiliary probability distribution $\qvt:=({q}_1,\ldots,{q}_{N})$ with ${q}_j=0$ once ${p}_j^*=0$ in the set $\psp^{N}$. Multiplying both sides of (\ref{eq:3.3}) by ${q}_j$ and summing over all $j$, 
\begin{align}
\nonumber
C_\li(\wuvt^\li)&\leq{\frac{1}{\li}\sum_{j=1}^{N}\sum_{i=1}^{M}q_j\wuvt^{\li}\left(\mathbf{y}_i|\mathbf{x}_j\right) {\log{\frac{\wuvt^{\li}\left(\mathbf{y}_i|\mathbf{x}_j\right)}{\sum_{j=1}^{N}{\pnvt_j^*\wuvt^{\li}\left(\mathbf{y}_i|\mathbf{x}_j\right)}}}}}.
\end{align}

Making use of the concentration set ${\mathcal{A}}$ in Definition~\ref{def:6}, we get~(\ref{eq:3.4}) where $\gamma_n\rightarrow 0$ as $\li\rightarrow\infty$. Moreover, the decomposition $\mathcal{A}=\bigcup_{k\in\mathcal{K}}\mathcal{A}_k$ (with $\mathcal{A}_k:=\mathcal{X}^\li\times\mathcal{B}_k$) yields~(\ref{eq:3.13}). Since logarithmic functions are concave and Eq.~(\ref{eq:3.11}) implies
\begin{align*}
\sum_{\left(\mathbf{x}_j,\mathbf{y}\right)\in {\mathcal{A}_k}}q_j f_k\left(\mathbf{y}|\mathbf{x}_j\right) =\sum_{\mathbf{y}\in\mathcal{B}_k}\sum_{j=1}^{N}q_j f_k\left(\mathbf{y}|\mathbf{x}_j\right)  = \sum_{j=1}^{N}q_j =1,
\end{align*}
applying Jensen's inequality to (\ref{eq:3.13}), it follows that
\begin{align}
\nonumber
&\sum_{k\in\mathcal{K}}\sum_{\left(\mathbf{x}_j,\mathbf{y}\right)\in {\mathcal{A}_k}}\frac{\wuvt^{\li}\left(\mathbf{y}|\mathbf{x}_j\right)}{f_k\left(\mathbf{y}|\mathbf{x}_j\right)}q_j f_k\left(\mathbf{y}|\mathbf{x}_j\right) {\log{\frac{\wuvt^{\li}\left(\mathbf{y}|\mathbf{x}_j\right)}{\sum_{j=1}^{N}{\pnvt_j^*\wuvt^{\li}\left(\mathbf{y}|\mathbf{x}_j\right)}}}}\\
\label{eq:3.7}
&\leq \max_{k\in\mathcal{K}}\log \bigg(\sum_{\left(\mathbf{x}_j,\mathbf{y}\right)\in {\mathcal{A}_k}}\frac{{q}_j \wuvt^{\li}\left(\mathbf{y}|\mathbf{x}_j\right)f_k\left(\mathbf{y}|\mathbf{x}_j\right)}{\sum_{j=1}^{N}{\pnvt_j^*\wuvt^{\li}\left(\mathbf{y}|\mathbf{x}_j\right)}}\bigg)
\end{align}
where the last inequality holds since Eq.~(\ref{eq:3.11}) guarantees that
\begin{align*}
\sum_{k\in\mathcal{K}}\max_{\left(\mathbf{x},\mathbf{y}\right)\in\mathcal{A}_k}\frac{\wuvt^{\li}\left(\mathbf{y}|\mathbf{x}\right)}{f_k\left(\mathbf{y}|\mathbf{x}\right)}\leq 1.
\end{align*}


Next, we set ${q}_j={\pnvt_j^*}$ for all $j\in\{1,\ldots,N\}$. The quantity inside the logarithm of (\ref{eq:3.7}) becomes 
\begin{align}
\label{eq:3.8}
&\sum_{\left(\mathbf{x}_j,\mathbf{y}\right)\in {\mathcal{A}_k}}\frac{{p}_j^*\wuvt^{\li}\left(\mathbf{y}|\mathbf{x}_j\right)f_k\left(\mathbf{y}|\mathbf{x}_j\right)}{\sum_{j=1}^{N}{\pnvt_j^*\wuvt^{\li}\left(\mathbf{y}|\mathbf{x}_j\right)}}\\
=&\sum_{\mathbf{y}\in\mathcal{B}_k}\frac{\sum_{j=1}^{N}{p}_j^*\wuvt^{\li}\left(\mathbf{y}|\mathbf{x}_j\right)f_k\left(\mathbf{y}|\mathbf{x}_j\right)}{\sum_{j=1}^{N}{\pnvt_j^*\wuvt^{\li}\left(\mathbf{y}|\mathbf{x}_j\right)}}\\
\leq& \sum_{\mathbf{y}\in {\mathcal{B}_k}}\max_{\mathbf{x}\in\mathcal{X}^\li}f_k\left(\mathbf{y}|\mathbf{x}\right).
\end{align}

Putting (\ref{eq:3.7}) and (\ref{eq:3.8}) into (\ref{eq:3.4}), for any concentration set ${\mathcal{A}}=\bigcup_{k\in\mathcal{K}}\mathcal{X}^\li\times\mathcal{B}_k$ and stochastic factors $\{f_k\}_k$,
\begin{align}
\nonumber
C_\li(\wuvt^\li)&\leq\overline{C}_\li(\wuvt^\li)\\
\label{eq:3.9}
&:={\frac{1}{\li}}\max_{k\in\mathcal{K}}\log\bigg( \sum_{\mathbf{y}\in {\mathcal{B}_k}}\max_{\mathbf{x}\in\mathcal{X}^\li}f_k\left(\mathbf{y}|\mathbf{x}\right)\bigg)+\gamma_\li. 
\end{align}




Note that the term $\gamma_\li$ is vanishing (in $\li$). Hence, for information stable channels, the general formula in (\ref{eq:lim}) implies that
\begin{align}
\nonumber
C\leq \liminf_{\li\rightarrow \infty}  \overline{C}_\li(\wuvt^\li).
\end{align}

This completes the proof of Theorem~\ref{thm:0}. \hfill $\square$

\subsection{Verification of the Tightness for the BEC and BSC}
\label{sec:3.3}
This section is devoted to verifying the tightness of the upper bound in Theorem~\ref{thm:0} on the BEC and the BSC. Denote by $p\in (0,1)$ the \emph{erasure/bit-flip probability}. Note that under the settings of the BEC$(p)$ and BSC$(p)$, we have the following realizations of the input and output spaces:
\begin{align*}
&\mathcal{X}_{\mathrm{BEC}}^\li=\mathcal{X}_{\mathrm{BSC}}^\li:=\left\{0,1\right\}^\li
\end{align*}
and
\begin{align*}
&\overline{\mathcal{Y}}_{\mathrm{BEC}}:=\left\{0,1,E\right\}^\li,\\
&\overline{\mathcal{Y}}_{\mathrm{BSC}}:=\left\{0,1\right\}^\li.
\end{align*}

Denote by $d_E\left(\cdot\right)$ the number of erasures (marked as $E$) in a given the Hamming distance and $d_H\left(\cdot,\cdot\right)$ the Hamming distance. We select the following concentration sets (and the corresponding decompositions) for the two types of channels respectively.

\begin{definition}
We consider the following concentration sets denoted by $\mathcal{A}^{\mathrm{BEC}}$ and $\mathcal{A}^{\mathrm{BSC}}$ for the BEC and BSC respectively according to Definition~\ref{def:6}:
\begin{align}
\label{eq:3.17}
&\mathcal{A}^{\mathrm{BEC}}:=\bigcup_{k\in\mathcal{K}_\varepsilon}\mathcal{A}^{\mathrm{BEC}}_k,\\
\label{eq:3.18}
&\mathcal{A}^{\mathrm{BSC}}:=\bigcup_{k\in\mathcal{K}_\varepsilon}\mathcal{A}^{\mathrm{BSC}}_k
\end{align}
where $\mathcal{A}^{\mathrm{BEC}}_k$ and  $\mathcal{A}^{\mathrm{BSC}}_k$ are defined as follows in agreement with Definition~\ref{def:7}:
\begin{align}
\nonumber
\mathcal{K}_\varepsilon&:=\left\lfloor\li \left(p-\varepsilon\right),\li \left(p+\varepsilon\right)\right\rceil,\\
\label{eq:3.23}
\mathcal{A}^{\mathrm{BEC}}_k&:=\Big\{\left(\mathbf{x},\mathbf{y}\right)\in\{0,1\}^\li\times\left\{0,1,E\right\}^\li: d_E\left(\mathbf{y}\right)=k\Big\},\\
\label{eq:3:24}
\mathcal{A}^{\mathrm{BSC}}_k&:=\Big\{\left(\mathbf{x},\mathbf{y}\right)\in\{0,1\}^\li\times\left\{0,1\right\}^\li: d_H\left(\mathbf{x},\mathbf{y}\right)=k\Big\}
\end{align}
and the sets $\mathcal{B}^{\mathrm{BEC}}_k$ and $\mathcal{B}^{\mathrm{BSC}}_k$ are defined as
\begin{align}
\label{eq:3.25}
\mathcal{B}^{\mathrm{BEC}}_k&:=\Big\{\mathbf{y}\in\left\{0,1,E\right\}^\li: d_E\left(\mathbf{y}\right)=k\Big\},\\
\label{eq:3.26}
\mathcal{B}^{\mathrm{BSC}}_k&:=\{0,1\}^\li
\end{align}
with some $\varepsilon$ satisfying $\min\{p,1-p\}>\varepsilon>0$.
\end{definition}

\begin{lemma}
	The \emph{concentration sets} $\mathcal{A}^{\mathrm{BEC}}$ and $\mathcal{A}^{\mathrm{BSC}}$  defined in above satisfy the conditions (\ref{eq:3.22})-(\ref{eq:3.15}) in Definition~\ref{def:6}.
\end{lemma}

\begin{proof}
For the BEC$(p)$, since each bit of the length-$\li$ input sequence is erased i.i.d., the number of erased bits is distributed according to Bernoulli$(\li,p)$. Therefore, fix input length $\li$, by the Chernoff bound~ (see \cite{chernoff1981note,mitzenmacher2005probability}), source $\ivt$, the probability of the output sequence being inside the concentration set $\mathcal{A}^{\mathrm{BEC}}$ defined in (\ref{eq:3.17}) is always bounded from below by $1-2\exp\left(-\li p\varepsilon^2/3\right)$. Moreover, the information densities corresponding to the outliers must be bounded. Thus, the concentration set $\mathcal{A}^{\mathrm{BEC}}$ satisfies the condition in (\ref{eq:3.15}). For the BSC, we have a trivial decomposition $\mathcal{A}^{\mathrm{BSC}}=\{0,1\}^\li\times\overline{\mathcal{Y}}_{\mathrm{BSC}}$, automatically guarantees the condition in (\ref{eq:3.15}).
\end{proof}

Furthermore, we associate the following stochastic factors to both channels.

\begin{definition}
The \emph{stochastic factors} for the BEC and BSC are defined as follows.
\begin{align}
	\label{eq:3.16}
	&f_k^{\mathrm{BEC}}\left(\mathbf{x},\mathbf{y}\right):=
	\begin{cases}
	{1}/{{\li\choose k}} \quad &\text{if } d_E\left(\mathbf{y}\right) =k\\
	0 &\text{otherwise}
	\end{cases}
	\end{align}
	and
	\begin{align}
	\label{eq:3.19}
	&f_k^{\mathrm{BSC}}\left(\mathbf{x},\mathbf{y}\right):=
	\begin{cases}
	{1}/{{\li \choose k}} \quad &\text{if } d_H\left(\mathbf{x},\mathbf{y}\right) =k\\
	0 &\text{otherwise}
	\end{cases}
	\end{align}
	with $1\leq k\leq \li$ satisfying
	\begin{align}
	\label{eq:3.21}
	\left\lfloor\li \left(p-\varepsilon\right)\right\rfloor \leq k\leq \left\lceil\li \left(p+\varepsilon\right)\right\rceil
	\end{align}
	for some $\varepsilon\in\left(0,\min\{p,1-p\}\right)$.
\end{definition}

\begin{lemma}
	The \emph{stochastic factors} $f_k^{\mathrm{BEC}}\left(\mathbf{x},\mathbf{y}\right)$ and $f_k^{\mathrm{BSC}}\left(\mathbf{x},\mathbf{y}\right)$  defined in above satisfy the conditions (\ref{eq:3.11})-(\ref{eq:3.12}) in Definition~\ref{def:7}.
\end{lemma}

\begin{proof}
We check the stochastic factors defined above satisfy the conditions in (\ref{eq:3.11})-(\ref{eq:3.12}). For the BEC$(p)$, plugging in (\ref{eq:3.23}) and (\ref{eq:3.16}),
\begin{align}
\nonumber
\sum_{\mathbf{y}\in{\mathcal{B}^{\mathrm{BEC}}_k}}f^{\mathrm{BEC}}_k\left(\mathbf{y}|\mathbf{x}\right) = \frac{\left|\mathcal{B}^{\mathrm{BEC}}_k\right|}{2^{\li-k}}\cdot \frac{1}{{\li\choose k}}&=1, \ \ \forall \ \mathbf{x}\in\mathcal{X}^{\li}, k\in\mathcal{K}.
\end{align}
Since $\wuvt^{\li}\left(\mathbf{y}|\mathbf{x}\right)=p^{k}\left(1-p\right)^{\li-k}$ for all $\left(\mathbf{x},\mathbf{y}\right)\in\mathcal{A}^{\mathrm{BEC}}_k$,
\begin{align*}
&\sum_{k\in\mathcal{K}}\max_{\left(\mathbf{x},\mathbf{y}\right)\in\mathcal{A}^{\mathrm{BEC}}_k}\frac{\wuvt^{\li}\left(\mathbf{y}|\mathbf{x}\right)}{f^{\mathrm{BEC}}_k\left(\mathbf{y}|\mathbf{x}\right)}\\
&=\sum_{k=\left\lfloor\li \left(p-\varepsilon\right)\right\rfloor}^{\left\lceil\li \left(p+\varepsilon\right)\right\rceil}{\li\choose k}p^{k}\left(1-p\right)^{\li-k}\in [0,1]
\end{align*}
showing that $f_k^{\mathrm{BEC}}\left(\mathbf{x},\mathbf{y}\right)$ are stochastic factors. 

For the the BSC$(p)$, similarly, based on the definitions in (\ref{eq:3.19}), since for each fixed $\mathbf{x}\in\{0,1\}^\li$, there are in total $\li\choose k$ many $\mathbf{y}\in\mathcal{B}_k^{\mathrm{BSC}}=\{0,1\}^\li$ satisfying $d_H\left(\mathbf{x},\mathbf{y}\right)=k$, it follows that
\begin{align}
\nonumber
\sum_{\mathbf{y}\in{\mathcal{B}_k^{\mathrm{BSC}}}}f_k^{\mathrm{BSC}}\left(\mathbf{y}|\mathbf{x}\right) &= {\li\choose k}\cdot \frac{1}{{\li\choose k}}=1, \ \ \forall \ \mathbf{x}\in\mathcal{X}^{\li}, k\in\mathcal{K}.
\end{align}
Moreover, $\wuvt^{\li}\left(\mathbf{y}|\mathbf{x}\right)=p^{k}\left(1-p\right)^{\li-k}$ for all $\left(\mathbf{x},\mathbf{y}\right)\in\mathcal{A}^{\mathrm{BSC}}_k$, clearly,
\begin{align}
0&\leq \sum_{k\in\mathcal{K}}\max_{\left(\mathbf{x},\mathbf{y}\right)\in\mathcal{A}^{\mathrm{BSC}}_k}\frac{\wuvt^{\li}\left(\mathbf{y}|\mathbf{x}\right)}{f^{\mathrm{BSC}}_k\left(\mathbf{y}|\mathbf{x}\right)}.
\end{align}

From the definition of the set $\mathcal{A}^{\mathrm{BSC}}_k$,
\begin{align*}
&\sum_{k\in\mathcal{K}}\max_{\left(\mathbf{x},\mathbf{y}\right)\in\mathcal{A}^{\mathrm{BSC}}_k}\frac{\wuvt^{\li}\left(\mathbf{y}|\mathbf{x}\right)}{f^{\mathrm{BSC}}_k\left(\mathbf{y}|\mathbf{x}\right)}\\
=& \sum_{k\in\mathcal{K}}\max_{\left(\mathbf{x},\mathbf{y}\right):d_H\left(\mathbf{x},\mathbf{y}\right)=k}\frac{\wuvt^{\li}\left(\mathbf{y}|\mathbf{x}\right)}{f^{\mathrm{BSC}}_k\left(\mathbf{y}|\mathbf{x}\right)}\\
=&\sum_{k=\left\lfloor\li \left(p-\varepsilon\right)\right\rfloor}^{\left\lceil\li \left(p+\varepsilon\right)\right\rceil}{\li\choose k}p^{k}\left(1-p\right)^{\li-k}\in [0,1]
\end{align*}
verifying that the stochastic factors in (\ref{eq:3.16})-(\ref{eq:3.19}) satisfy the conditions in Definition~\ref{def:6} and Definition~\ref{def:7}.
\end{proof}

Based on the concentration sets $\mathcal{A}^{\mathrm{BEC}}$ and $\mathcal{A}^{\mathrm{BSC}}$ defined in (\ref{eq:3.17})-(\ref{eq:3.26}) and the stochastic factors in (\ref{eq:3.16})-(\ref{eq:3.19}), the ML upper bound in Theorem~\ref{thm:0} is tight, as stated in the following theorem.
\begin{theorem}[Tightness for the BEC and BSC] Let $p\in(0,1)$ be the \emph{erasure/bit-flip probability}.
The ML upper bound in Theorem~\ref{thm:0} is tight for the BEC and BSC, \emph{i.e.,}
\begin{align*}
&\liminf_{\li\rightarrow\infty}{\frac{1}{\li}}\max_{k\in\mathcal{K}_{\varepsilon}}\log\Big( \sum_{\mathbf{y}\in {\mathcal{B}^{\mathrm{BEC}}_k}}\max_{\mathbf{x}\in\mathcal{X}_{\mathrm{BEC}}^\li}f^{\mathrm{BEC}}_k\left(\mathbf{y}|\mathbf{x}\right)\Big) = 1-p,\\
&\liminf_{\li\rightarrow\infty}{\frac{1}{\li}}\max_{k\in\mathcal{K}_{\varepsilon}}\log\Big( \sum_{\mathbf{y}\in {\mathcal{B}^{\mathrm{BSC}}_k}}\max_{\mathbf{x}\in\mathcal{X}_{\mathrm{BSC}}^\li}f^{\mathrm{BSC}}_k\left(\mathbf{y}|\mathbf{x}\right)\Big) = 1-h\left(p\right)
\end{align*}
where $h(p):=-p\log p -\left(1-p\right)\log \left(1-p\right)$ denotes the binary entropy.
\end{theorem}

\begin{proof}
	Note that the parameter $\varepsilon>0$ can be arbitrarily small, taking $k=\lceil\li p\rceil$ and applying Theorem~\ref{thm:0}, the capacity of the BEC denoted by $C_{\mathrm{BEC}}(p)$ satisfies
	\begin{align*}
	C_{\mathrm{BEC}}(p)\leq& \liminf_{\li\rightarrow\infty} {\frac{1}{\li}}\log\bigg(\sum_{\mathbf{y}\in {\mathcal{B}^{\mathrm{BEC}}_{\lceil\li p\rceil}}}\max_{\mathbf{x}\in\mathcal{X}_{\mathrm{BEC}}^\li}f^{\mathrm{BEC}}_{\lceil\li p\rceil}\left(\mathbf{y}|\mathbf{x}\right)\bigg)\\
	=&\liminf_{\li\rightarrow\infty} {\frac{1}{\li}}\log\bigg({\left| \mathcal{B}_{\lceil\li p\rceil}^{\mathrm{BEC}}\right|}{\big/}{{\li\choose \lceil\li p\rceil}}\bigg).
	\end{align*}
	
	Putting $\left| \mathcal{B}_{\lceil\li p\rceil}^{\mathrm{BEC}}\right| =2^{\li\left(1-p\right)}{\li\choose \lceil\li p\rceil}$ into above,
	\begin{align}
	\nonumber
	C_{\mathrm{BEC}}(p)&\leq\liminf_{\li\rightarrow\infty} {\frac{1}{\li}}\log \left(2^{\li\left(1-p\right)}{\li\choose \lceil\li p\rceil}{\big/}{{\li\choose \lceil\li p\rceil}} \right)\\
	\label{eq:3.27}
	&=1-p.
	\end{align}
	
	Furthermore, for the BSC, the capacity $C_{\mathrm{BSC}}(p)$ is bounded from above as
	\begin{align*}
	C_{\mathrm{BSC}}(p)\leq& \liminf_{\li\rightarrow\infty} {\frac{1}{\li}}\log\bigg( \sum_{\mathbf{y}\in {\mathcal{B}^{\mathrm{BSC}}_{\lceil\li p\rceil}}}\max_{\mathbf{x}\in\mathcal{X}_{\mathrm{BSC}}^\li}f^{\mathrm{BSC}}_{\lceil\li p\rceil}\left(\mathbf{y}|\mathbf{x}\right)\bigg)\\
	=&\liminf_{\li\rightarrow\infty} {\frac{1}{\li}}\log\left(\left| \mathcal{B}_{\lceil\li p\rceil}^{\mathrm{BSC}}\right|{\big/}{{\li\choose \lceil\li p\rceil}}\right).
	\end{align*}
	
	Since $\left| \mathcal{B}_{\lceil\li p\rceil}^{\mathrm{BSC}}\right| =2^{\li}$, it follows that
	\begin{align}
	\label{eq:3.28}
	C_{\mathrm{BSC}}(p)\leq\liminf_{\li\rightarrow\infty} {\frac{1}{\li}}\log \left(2^{\li}{\big/}{{\li\choose \lceil\li p\rceil}} \right)=1-h(p).
	\end{align}
\end{proof}


The theorem above indicates that the ML upper bound in Theorem~\ref{thm:0} is actually tight for some memoryless channels, \emph{e.g.},  the BEC$(p)$ and BSC$(p)$. In next section, we analyze the binary deletion channel (BDC) as an example for channels with memory, and show that the ML upper bound is capable of providing a nontrivial and explicit approximation for the capacity $C(\delep)$ ($\delep$ denotes the deletion probability) assuming Hypothesis~\ref{hypo:2} in Section~\ref{sec:4.3}. 



\section{Binary Deletion Channel}
\label{sec:4}

For brevity, we consider specifically the binary deletion channel (BDC), though the approach generalizes to arbitrary alphabet sizes. The following section is  devoted to summarizing related work on finding the capacity $C(\delep)$ of the BDC. In particular, we focus on the existing upper bounds on $C(\delep)$ and the known asymptotic results. The survey by Mitzenmacher~\cite{mitzenmacher2009survey} elucidates critical problems, useful techniques and further applications in a more comprehensive way. The recent paper by Cheraghchi also provides a decent summary of the state-of-the-art literature~\cite{cheraghchi2017capacity}. Before moving to the contexts, we first give a brief summary of known bounds on $C(\delep)$, the capacity of the BDC, at the risk of missing much of the literature.



\subsection{Previous Work}

\subsubsection{Existing Upper Bounds}

Recently, Cheraghchi in~\cite{cheraghchi2017capacity} gave an explicit and concise upper bound on $C(d)$ such that $C(d)\leq1-d\log\left({4}/{\phi}\right)$ for $d< {1}/{2}$ and $C(d)\leq(1-d)\log \phi$ for $d\geq {1}/{2}$ where $\phi:=\left({1+\sqrt{5}}\right)/{2}$ is the golden ratio. The bound was obtained by first deriving an upper bound on $C(1/2)$; then applying the fact that $C(d)$ is convex as showed in \cite{rahmati2015upper}.

Running the BAA (cf.~\cite{arimoto1972algorithm,blahut1972computation}) up to $\li=17$,  tighter numerical upper bounds were provided in~\cite{fertonani2010novel} improving the upper bounds in~\cite{diggavi2007capacity} for a wide range of the deletion probability $d$. They proved that increasing the dimension $\li$ in the BAA always provides a better upper bound on $C(\delep)$. The convexity of $C(d)$ in~\cite{rahmati2015upper} can be used to  tighten the bounds in~\cite{fertonani2010novel}.
However, the space complexity of BAA is exponential in $\li$, prohibiting obtaining better bounds by trying larger dimensions. Following previous literature (\cite{dalai2011new}), we sometimes replace the maximized finite-length mutual information $C_\li(\wuvt^\li)$ with $C_\li(\delep)$ since the quantity is determined entirely by $\delep$ and $\li$ in particular for deletion channels. Denote by $C_{\li,T}^{\mathrm{BAA}}(d)$ the approximation of $C_{\li}(\delep)$ using the BAA with $T$ iterations. {\it A priori}, if one could use the BAA to obtain the value of $C_\li(\delep)$ precisely, one would be able to get a ${\log(\li+1)}/{\li}$ additive approximation to the capacity $C(\delep)$ (for $\li\leq 17$). In~\cite{fertonani2010novel}, essentially this approach is followed to obtain numerical solutions for~(\ref{eq:2.2}), for $\li\leq 17$. 

However, this approach has several limitations. On the one hand, in direct implementation of BAA, as $\li$ grows, it becomes computationally intractable even to store the variables to be computed. One the other hand, as BAA is itself an iterative algorithm attempting to solve the non-convex optimization problem~(\ref{eq:2.2}), and to the best of our knowledge there are no guarantees on how quickly $C_{\li,T}^{\mathrm{BAA}}(d)$ converges to $C_n(d)$ as a function of the number of iterations $t$.

In fact, in ~\cite{dobrushin1967shannon}, Dobrushin showed the following quantitative bound on $C_\li(\delep)$ providing a ${\log(\li+1)}/{\li}$-additive approximation to $C(\delep)$ (see also, ~\cite{dalai2011new, kanoria2013optimal}. A tighter bound can be found in~\cite{kalai2010tight}):
\begin{align}
\label{eq:2.0}
{\ C_\li(\delep)}- \frac{\log(\li+1)}{\li}\leq C(\delep) \leq {\ C_\li(\delep)}. 
\end{align}



\subsubsection{Asymptotic Results}

Besides the upper bounds, in \cite{kanoria2013optimal} Kanoria and Montanari gave a polynomial expression of channel capacity which is optimal with a residual term $O(d^{3-\varepsilon})$ through the optimality of sources with i.i.d. coordinates distributed as Bernoulli~$(1/2)$ when $\delep=0$ and a perturbed version when $\delep$ is slightly larger than $0$.  Meanwhile for the regime $\delep\rightarrow 1$, Dalai in \cite{dalai2011new} provided an  asymptotic result $\lim_{\delep\rightarrow 1}C(\delep)/(1-\delep)\leq 0.4143$ with constant surpassing the upper bound $0.49$ given in \cite{fertonani2010novel} by Fertonani and Duman.  Furthermore, \cite{drmota2012mutual} studied the mutual information for deletion channels concerning both general sources and memoryless sources. In particular, their results for memoryless sources coincide with Kanoria and Montanari's expansion of mutual information~\cite{kanoria2013optimal} as $\delep\rightarrow 0$. 

In the sequel, we apply the ML upper bound in Theorem~\ref{thm:0} to derive approximations for the capacity $C(\delep)$ of the BDC.

\subsection{Capacity Upper Bound via Theorem~\ref{thm:0}}
\label{sec:4.2}
\subsubsection{Concentration Set and Stochastic Factors}
\label{sec:4.2.1}
The input and output spaces of the BDC are
\begin{align*}
&\mathcal{X}_{\mathrm{BDC}}^\li:=\left\{0,1\right\}^\li
\end{align*}
and
\begin{align*}
&\overline{\mathcal{Y}}_{\mathrm{BDC}}:=\bigcup_{1\leq \lo\leq \li}\left\{0,1\right\}^\lo.
\end{align*}

We consider the following concentration set and the corresponding decompositions for the BDC:

\begin{definition}[Concentration Set for the BDC]
	\label{def:8}
	For the BDC, we define the following \emph{concentration set} according to Definition~\ref{def:6}:
	\begin{align}
	\label{eq:4.2}
	\mathcal{A}^{\mathrm{BDC}}:=\bigcup_{k\in\mathcal{K}}\mathcal{A}^{\mathrm{BDC}}_k
	\end{align}
	with the following decomposition in agreement with Definition~\ref{def:7}:
	\begin{align}
	\nonumber
	\mathcal{K}_{\varepsilon}&:=\left\lfloor\li \left(\delep-\varepsilon\right),\li \left(\delep+\varepsilon\right)\right\rceil,\\
	\nonumber
	\mathcal{A}^{\mathrm{BDC}}_k&:=\{0,1\}^\li\times\left\{0,1\right\}^{\li-k},\\
	\label{eq:4.5}
	\mathcal{B}^{\mathrm{BDC}}_k&:=\{0,1\}^{\li-k}.
	\end{align}
\end{definition}

As $\li$ grows without bound, standard concentration inequalities (for instance, Chernoff bound) imply that the length of the output sequence $\lo$ is tightly concentrated around the ``typical length'' $\li(1-\delep)$ in probability. 
Using the bounds in (\ref{eq:2.0}),~\cite{dalai2011new} showed that $C(\delep)$ is continuous (the continuity can be verified via various approaches, for instance, the information spectrum method~\cite{koga2013information}). Leveraging the continuity,  Theorem~$1$ in~\cite{dalai2011new} (\emph{cf.}~\cite{fertonani2010novel,kanoria2013optimal}) proved that it is sufficient to consider output sequences with lengths in $\mathcal{K}_{\varepsilon}=\left\lfloor\li \left(\delep-\varepsilon\right),\li \left(\delep+\varepsilon\right)\right\rceil$. In detail, Dalai showed that the following lemma holds.
\begin{lemma}[Theorem 1~\cite{dalai2011new}]
	\label{lemma:1}
For any $0<\varepsilon\leq \min\{\delep,1-\delep\}$,
	 \begin{align*}
	 &\liminf_{\li\rightarrow\infty}\frac{1}{\li}\mathbbm{E}\left[\infd_{\invt^\li, \wuvt^\li}{(\invt^\li;\ouvt\left(\invt^\li\right))}\big | \left(\mathbf{X}^\li,\mathbf{Y}\right)\in\mathcal{A}^{\mathrm{BDC}}\right]\\
	  =&\liminf_{\li\rightarrow\infty}C_\li\left(\wuvt^\li\right) = C(\delep).
	 \end{align*}
\end{lemma}

The lemma above validates  that $\mathcal{A}^{\mathrm{BDC}}$ is a concentration set according to Definition~\ref{def:6}.

\begin{lemma}
	The \emph{concentration set} $\mathcal{A}^{\mathrm{BDC}}$ defined in above satisfies the conditions (\ref{eq:3.22})-(\ref{eq:3.15}) in Definition~\ref{def:6}.
\end{lemma}
\begin{proof}
Since by definition the concentration set $\mathcal{A}^{\mathrm{BDC}}$ contains all received codewords with lengths in the range $\left\lfloor\li \left(\delep-\varepsilon\right),\li \left(\delep+\varepsilon\right)\right\rceil$, standard concentration inequalities (\emph{e.g.} Chernoff bound) guarantees the condition in (\ref{eq:3.22}), and
\begin{align}
\nonumber
&\liminf_{\li\rightarrow\infty}\mathbbm{P}\left(\left(\mathbf{X}^\li,\mathbf{Y}\right)\in\mathcal{A}^{\mathrm{BDC}}\right)=1.
\end{align}

Moreover, using Lemma~\ref{lemma:1} above, directly,
\begin{align*}
&\limsup_{\li\rightarrow\infty}\mathbbm{E}\left[\left|\frac{\infd_{\invt^\li, \wuvt^\li}{(\invt^\li;\ouvt\left(\invt^\li\right))}}{\li C_\li\left(\wuvt^\li\right)}-1\right|\Big | \left(\mathbf{X}^\li,\mathbf{Y}\right)\in\mathcal{A}^{\mathrm{BDC}}\right]\\
&=0.
\end{align*}

\end{proof}
 
Before proceeding to the corresponding stochastic factors, it is helpful to introduce a quantity pertinent to relationships between length-$\li$ input sequences ${x}$ and length-$\lo$ output sequences ${\mathbf{y}}$.

\begin{definition}[Deletion Pattern]
	\label{def:3}
	A \emph{deletion pattern} from a length-$\li$ input sequence ${\mathbf{x}}$ to a length-$\lo$ output sequence ${\mathbf{y}}$ is a binary sequence denoted by $\mathbf{d}\in\{0,1\}^\li$. If the $i$-th coordinate $d_i=1$, the corresponding $i$-th coordinate $x_i$ in $\mathbf{x}$ is deleted; otherwise $d_i=0$ implies that $x_i$ is kept.
\end{definition} 

Below we define a function computing the number of possible ways that a fixed length-$\li$ sequence is deleted to form a fixed shorter length-$\lo$ sequence. Previous studies have been focusing on similar quantities, Drmota \emph{et al.} defined a similar quantity as the number of occurrences of a shorter sequence in a longer sequence, Liron and langberg characterized the number of subsequences obtained from a fixed length-$\li$ sequence via deletions~\cite{liron2015characterization,drmota2012mutual}, to name just a few. 

\begin{definition}[Number of Deletion Patterns\footnote{This counting function is alternatively called \emph{hidden pattern matching function} in~\cite{bourdon2002generalized} using the terminology in statistics.}]
	\label{def:2}
	We define the {{\it number of deletion patterns}} as a quantity $\#\left(\mathbf{x},\mathbf{y}\right)\in\left\{1,\ldots, {\li\choose \lo}\right\}$ counting the number of distinct deletion patterns from an input ${\mathbf{x}}\in\{0,1\}^{\li}$ to an output ${\mathbf{y}}\in\{0,1\}^{\lo}$.
\end{definition} 

Over the years it has been repeatedly noted that the number of deletion patterns plays an important role in finding the capacity $C(\delep)$. Part of the reason is that the number of deletion patterns can be regarded as a ``normalized version'' of the transition probability $\wuvt^{\li}\left({\mathbf{y}}|{\mathbf{x}}\right)$, as the following remark explains .


\begin{remark}
	\label{remark:1}
	Denote by $\wuvt^{\li}\left(\cdot|\cdot\right)$ the corresponding stochastic matrix for the BDC (with block length $\li$). Since the probability of a particular deletion pattern of weight $\li-\lo$ occurring equals $(1-\delep)^{\lo}\delep^{\li-\lo}$, hence
	\begin{align}
	\label{eq:4.7}
	\#\left({\mathbf{x}},{\mathbf{y}}\right)=\frac{\wuvt^{\li}\left({\mathbf{y}}|{\mathbf{x}}\right)}{(1-\delep)^{\lo}\delep^{\li-\lo}}.
	\end{align}
	The fact that $\#\left({\mathbf{y}},{\mathbf{x}}\right)$ is a scaled version of $\wuvt^{\li}\left({\mathbf{y}}|{\mathbf{x}}\right)$ (corresponding to a conditional probability distribution for the input $\mathbf{x}$ being mapped to the fixed output $\mathbf{y}$) takes on operational significance later. For instance, it ensures us to define stochastic functions. In Section~\ref{sec:4.5}, on the other hand, we utilize the operational meaning of $\#\left({\mathbf{y}},{\mathbf{x}}\right)$  to derive explicit bounds on $C(\delep)$.
\end{remark}

\begin{definition}
	\label{def:9}
	The corresponding stochastic factors for the BDC are set to be
	\begin{align}
	\label{eq:4.6}
	&f_k^{\mathrm{BDC}}\left(\mathbf{x},\mathbf{y}\right):=
	\begin{cases}
	{	\#\left({\mathbf{x}},{\mathbf{y}}\right)}/{{\li\choose k}} \quad &\text{if } \mathbf{y}\in\{0,1\}^{\li-k}\\
	0 &\text{otherwise}
	\end{cases}.
	\end{align}
\end{definition}

It remains to check the validity of the stochastic factors $f_k^{\mathrm{BDC}}\left(\mathbf{x},\mathbf{y}\right)$. We first claim that they satisfy the definition of stochastic factors.
\begin{lemma}
	The \emph{stochastic factors} $f_k^{\mathrm{BEC}}\left(\mathbf{x},\mathbf{y}\right)$ and $f_k^{\mathrm{BSC}}\left(\mathbf{x},\mathbf{y}\right)$  defined in above satisfy the conditions (\ref{eq:3.11})-(\ref{eq:3.12}) in Definition~\ref{def:7}.
\end{lemma}

\begin{proof}
	Plugging in (\ref{eq:4.5}) and (\ref{eq:4.6}),
	\begin{align}
	\nonumber
	&\sum_{\mathbf{y}\in{\mathcal{B}^{\mathrm{BDC}}_k}}f^{\mathrm{BDC}}_k\left(\mathbf{y}|\mathbf{x}\right) =\sum_{\mathbf{y}\in\{0,1\}^{\li-k}}\frac{	\#\left({\mathbf{x}},{\mathbf{y}}\right)}{{\li\choose k}}  =1, \\
	& \forall \ \mathbf{x}\in\mathcal{X}^{\li}, k\in\mathcal{K}_\varepsilon.
	\end{align}
	
	Considering (\ref{eq:4.7}),
	\begin{align} &\sum_{k\in\mathcal{K}}\max_{\left(\mathbf{x},\mathbf{y}\right)\in\mathcal{A}^{\mathrm{BDC}}_k}\frac{\wuvt^{\li}\left(\mathbf{y}|\mathbf{x}\right)}{f^{\mathrm{BDC}}_k\left(\mathbf{y}|\mathbf{x}\right)}\\
	=&\sum_{k=\left\lfloor\li \left(p-\varepsilon\right)\right\rfloor}^{\left\lceil\li \left(p+\varepsilon\right)\right\rceil}{\li\choose k}p^{k}\left(1-p\right)^{\li-k}\in [0,1]
	\end{align}
	showing that $f_k^{\mathrm{BDC}}\left(\mathbf{x},\mathbf{y}\right)$ are stochastic factors.
\end{proof}


Making use of the concentration set and stochastic factors constructed in (\ref{eq:4.5}) and (\ref{eq:4.6}), and substituting them into the ML upper bound in Theorem~\ref{thm:0}, the following upper bound on $C(\delep)$ holds.
\begin{theorem}
	\label{thm:1}
	For all $\delep\in (0,1)$ and $\lo:=\lceil\li(1-\delep)\rceil$, the capacity $C(\delep)$ of the BDC is bounded from above by
	\begin{align}
	\label{eq:4.8}
	C(\delep)\leq\overline{C}(\delep):=\liminf_{\li\rightarrow\infty}\overline{C}_{\li}(\delep)- h(\delep)
	\end{align}
where $h(d)=-\delep\log \delep-\left(1-\delep\right)\log\left(1-\delep\right)$ denotes the binary entropy and
\begin{align}
\label{eq:4.10}
\overline{C}_{\li}(\delep):=\frac{1}{\li}\log\bigg(\sum_{\mathbf{y}\in\{0,1\}^{\lo}}\max_{\mathbf{x}\in\{0,1\}^{\li}}\#\left({\mathbf{x}},{\mathbf{y}}\right)\bigg).
\end{align}
\end{theorem}

\begin{proof}
Since $0<\varepsilon\leq \min\{\delep,1-\delep\}$ is arbitrary, taking the maximizing $k=\lo:=\lceil\li(1-\delep)\rceil$ and applying Theorem~\ref{thm:0}, the capacity of the BDC satisfies
\begin{align}
\nonumber
C(\delep)\leq& \liminf_{\li\rightarrow\infty} {\frac{1}{\li}}\log\bigg( \sum_{\mathbf{y}\in {\mathcal{B}^{\mathrm{BDC}}_{\lo}}}\max_{\mathbf{x}\in\mathcal{X}_{\mathrm{BDC}}^\li}f^{\mathrm{BDC}}_{\lo}\left(\mathbf{y}|\mathbf{x}\right)\bigg).
\end{align}

Applying the definitions of the sets $\mathcal{B}^{\mathrm{BDC}}_{\lo}$ and the stochastic factors $f^{\mathrm{BDC}}_{\lo}$ (in Definition~\ref{def:8} and~\ref{def:9}),
\begin{align}
\label{eq:4.9}
C(\delep)\leq&\liminf_{\li\rightarrow\infty} {\frac{1}{\li}}\log\bigg( \sum_{\mathbf{y}\in \{0,1\}^{\lo}}\max_{\mathbf{x}\in\{0,1\}^{\li}}\frac{\#\left({\mathbf{x}},{\mathbf{y}}\right)}{{\li\choose \lo}}\bigg)\\
\nonumber
=&\liminf_{\li\rightarrow\infty} {\frac{1}{\li}}\log\bigg( \sum_{\mathbf{y}\in \{0,1\}^{\lo}}\max_{\mathbf{x}\in\{0,1\}^{\li}}{\#\left({\mathbf{x}},{\mathbf{y}}\right)}\bigg) \\
\nonumber
\quad &+\liminf_{\li\rightarrow\infty} {\frac{1}{\li}}\log\frac{1}{{\li\choose \lo}}\\
\nonumber
=&\overline{C}_{\li}\left(\delep\right)-h(\delep),
\end{align}
which gives the desired bound in (\ref{eq:4.8}).
\end{proof}

\begin{remark}
Note that for any $\mathbf{x}\in\{0,1\}^\li$ and $\mathbf{y}\in\{0,1\}^\lo$, it always holds that
\begin{align*}
0\leq \#\left({\mathbf{x}},{\mathbf{y}}\right)\leq {\li\choose \lo}.
\end{align*}

Putting this into (\ref{eq:4.9}), we recover the trivial upper bound $C(\delep)\leq 1-\delep$.
\end{remark}

\begin{figure}[h!]
	\centering
	\includegraphics[scale=0.45]{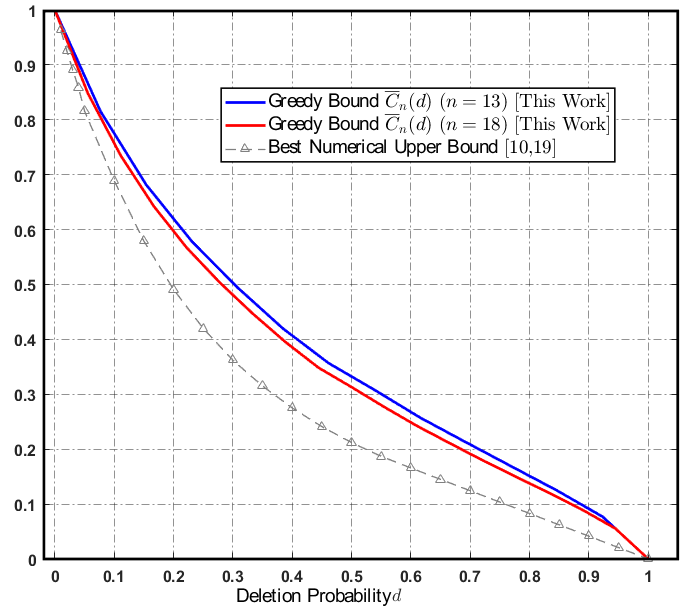}
	\caption{The ML upper bounds (solid, blue and red) $\overline{C}_{\li}(\delep)$ for BDC with block-length $\li=18$ and $\li=13$, together with the (convexified) numerical estimate of the capacity-proxy $C_\li(\delep)$ (dashed and marked black) for $\li=17$. The lower curve (dashed gray) is also known as the best numerical upper bounds provided in~\cite{rahmati2015upper}.}
	\label{fig:val}
	\medskip
	\hrule
\end{figure}


\subsubsection{Experimental Results}

We implement the BAA up to $\li=18$ to compare $\overline{C}_{\li}(\delep)$,  ${C}_{\li}(\delep)$ and the best known numerical bounds on $C(\delep)$. See Fig.~\ref{fig:val}.


In next section we make use of the operational meaning of the number of deletion patterns and analyze the quantity $\overline{C}_{\li}(\delep)$ in a more careful way. This allows us to derive upper bounds on $C(\delep)$ based on an approximation ratio of a combinatorial problem  defined later in Section~\ref{sec:4.2}.

\begin{table*}[h!]
	\centering
	\setlength\extrarowheight{1.35pt}
	\begin{tabular}{|c c c | c c c|} 
		\firsthline
		& \multicolumn{2}{c}{Exact} & \multicolumn{2}{|c}{ Approximation} & \multicolumn{1}{c|}{Duplication}  \\ 
		\cline{2-5} 
		$\mathbf{y}$ & $\mathbf{x}$ & $\overline{\#}\left({\mathbf{y}}\right)$ & $\mathbf{x}_{\mathrm{dup}}$ & $\overline{\#}_{\mathrm{dup}}\left({\mathbf{y}}\right) $ & Ratio \\ [0.5ex] 
		\hline\hline
		$0000$ & $00000000$ & $70$ & $00000000$ & $70$ & $1$ \\ 
		$0001$ & $00000011$ & $40$ & $00000011$ & $40$ & $1$ \\
		$0010$ & $00001100$ & $24$ & $00001100$ & $24$ & $1$ \\ 
		$0011$ & $00001111$ & $36$ & $00001111$ & $36$ & $1$ \\
		$0100$ & $00110000$ & $24$ & $00110000$ & $24$ & $1$ \\ [1ex] 
		$0101$ & $\mathbf{00101011}$ & $16$ & $\mathbf{00110011}$ & $16$ & $1$ \\ [1ex] 
		$0110$ & $00111100$ & $24$ & $00111100$ & $24$ & $1$ \\
		$0111$ & $00111111$ & $40$ & $00111111$ & $40$ & $1$\\
		\hline 
	\end{tabular}
	\vspace{3pt}
	\caption{A table showing the corresponding ratios given $\lo=4$ and $\li=8$. When $\lo\leq 4$, the ratio is always $1$.}
	\label{table:1}
\end{table*}

\begin{table*}[h!]
	\centering
	\setlength\extrarowheight{1.35pt}
	\begin{tabular}{|c c c | c c c|} 
		\firsthline
		& \multicolumn{2}{c}{Exact} & \multicolumn{2}{|c}{ Approximation} & \multicolumn{1}{c|}{Duplication} \\ 
		\cline{2-5} 
		$\mathbf{y}$ & $\mathbf{x}$ & $\overline{\#}\left({\mathbf{y}}\right)$ & $\mathbf{x}_{\mathrm{dup}}$ & $\overline{\#}_{\mathrm{dup}}\left({\mathbf{y}}\right) $ & Ratio \\ [0.5ex] 
		\hline\hline
		$0000000$ & $00000000000000$ & $3432$ & $00000000000000$ & $3432$ & $1$ \\ 
		$0000001$ & $00000000000011$ & $1848$ & $00000000000011$ & $1848$ & $1$ \\
		$0000010$ & $00000000001100$ & $1008$ & $00000000001100$ & $1008$ & $1$ \\ 
		$0000011$ & $00000000001111$ & $1512$ & $00000000001111$ & $1512$ & $1$ \\
		$0000100$ & $00000000110000$ & $840$ & $00000000110000$ & $840$ & $1$ \\ 
		$0000101$ & $\mathbf{00000000101011}$ & $602$ & $\mathbf{00000000110011}$ & $560$ & $0.93023.$ \\
		$0000110$ & $00000000111100$ & $840$ & $00000000111100$ & $840$ & $1$ \\
		$0000111$ & $00000000111111$ & $1400$ & $00000000111111$ & $1400$ & $1$\\
		$0000111$ & $00000000111111$ & $1400$ & $00000000111111$ & $1400$ & $1$\\
		$\vdots$ & 	$\vdots$ & 	$\vdots$ & 	$\vdots$ & 	$\vdots$ & 	$\vdots$\\
		$0001010$ & $\mathbf{00000000111111}$ & $396$ & $\mathbf{00000011001100}$ & $320$ & $0.80808.$\\
		$0001011$ & $\mathbf{00000010101111}$ & $530$ & $\mathbf{00000011001111}$ & $480$ & $0.90566.$\\
		$0010100$ & $\mathbf{00001010101000}$ & $351$ & $\mathbf{00001100110000}$ & $288$ & $0.90566.$\\
		$0010101$ & $\mathbf{00001010101011}$ & $270$ & $\mathbf{00001100110011}$ & $192$ & $0.71111.$\\
		$0010110$ & $\mathbf{00001010111100}$ & $312$ & $\mathbf{00001100111100}$ & $288$ & $0.92308.$\\
		$0010111$ & $\mathbf{00001010111111}$ & $530$ & $\mathbf{00001100111111}$ & $480$ & $0.90566.$\\
		$0011010$ & $\mathbf{00001111010100}$ & $300$ & $\mathbf{00001111001100}$ & $288$ & $0.96$\\
		$0100101$ & $\mathbf{00110000101011}$ & $200$ & $\mathbf{00110000110011}$ & $192$ & $0.96$\\
		$0101000$ & $\mathbf{01010101000000}$ & $396$ & $\mathbf{00110011000000}$ & $320$ & $0.80808.$\\
		$0101001$ & $\mathbf{01010101000011}$ & $231$ & $\mathbf{00110011000011}$ & $192$ & $0.83117.$\\ [1ex]
		$0101010$ & $\mathbf{00101010101010}$ & $204$ & $\mathbf{00110011001100}$ & $128$ & $\mathbf{0.62745.}$\\ [1ex]
		$0101011$ & $\mathbf{00101010101111}$ & $270$ & $\mathbf{00110011001111}$ & $192$ & $0.71111.$\\
		$0101100$ & $\mathbf{00101011110000}$ & $300$ & $\mathbf{00110011110000}$ & $288$ & $0.96$\\
		$0101101$ & $\mathbf{00101011110011}$ & $200$ & $\mathbf{00110011110011}$ & $192$ & $0.96$\\
		$0101110$ & $\mathbf{00101011111100}$ & $340$ & $\mathbf{00110011111100}$ & $320$ & $0.94118.$\\
		$0101111$ & $\mathbf{00101011111111}$ & $602$ & $\mathbf{00110011111111}$ & $560$ & $0.93023.$\\
		$0110100$ & $\mathbf{00111101010000}$ & $312$ & $\mathbf{00111100110000}$ & $280$ & $0.89743.$\\
		$0110101$ & $\mathbf{00111101010101}$ & $231$ & $\mathbf{00111100110011}$ & $192$ & $0.83117.$\\
		$0111010$ & $\mathbf{00111111010100}$ & $340$ & $\mathbf{00111111001100}$ & $320$ & $0.94118.$\\
		\hline 
	\end{tabular}
	\vspace{3pt}
	\caption{An incomplete list comparing the maximal number of deletion patterns with the duplication approximations for $\lo=7$ and $\li=14$. The table contains all ratios that are not equal to one. The smallest duplication ratio $\gamma\left(\li,\digamma\right)$ is obtained at the ``flipping sequence''  $\mathbf{y}_{\mathrm{flip}}=0101010$.}
	\label{table:2}
\end{table*}

\subsection{Maximal Number of Deletion Patterns}
\label{sec:4.3}


The remaining context of this paper is dedicated to approximate the terms in (\ref{eq:4.9}), which is summarized as the following combinatorial problem.

\subsubsection{The Maximum Deletion Matching problem}
\begin{definition}[MDM Problem]
	\label{def:4}
The \emph{maximum deletion matching (MDM) problem} is to solve the following. Let $\li,\lo\in\mathbbm{N}^{+}$ with $1\leq\lo\leq\li$. Given an arbitrary length-$\lo$ binary sequence $\mathbf{y}\in\{0,1\}^\lo$, the goal is to find the maximum corresponding length-$\li$ binary sequence $\mathbf{x}$ such that
\begin{align*}
\overline{\mathbf{x}}\left(\mathbf{y}\right):=\argmax_{\mathbf{x}\in\{0,1\}^{\li}}\#\left({\mathbf{x}},{\mathbf{y}}\right)
\end{align*}
where $\#\left({\mathbf{x}},{\mathbf{y}}\right)$ denotes the number of deletion patterns of generating $\mathbf{y}$ from $\mathbf{x}$ defined in Definition~\ref{def:2}. For notational convenience, write the \emph{maximal number of deletion patterns}
$$\overline{\#}\left({\mathbf{y}}\right):=\max_{\mathbf{x}\in\{0,1\}^{\li}}\#\left({\mathbf{x}},{\mathbf{y}}\right) = \#\left(\overline{\mathbf{x}}\left(\mathbf{y}\right),{\mathbf{y}}\right).$$
\end{definition}

\subsubsection{Run-length Representation}

One way to approximate $\overline{\#}\left({\mathbf{y}}\right)$ and find approximation ratio of the MDM problem is to consider consecutive bits in $\mathbf{y}$ as a  ``run'' and jointly a distribution of run-lengths. Although encoding each sequence $\mathbf{y}$ to the run-length representation suffers a loss of information (for instance, the ordering of runs is no longer kept in the run-length representation), it offers a concise approach to describe a binary sequence. Reprising the definitions from previous work~\cite{kirsch2010directly,kanoria2013optimal,liron2015characterization}, we consider the follows.

First, we associate each length-$\lo$  binary sequence ${\mathbf{y}}$ with an integer sequence embedding the information of number of consecutive bits of ${\mathbf{y}}$. We call a subsequence $y_{i},\ldots,y_{i+\ell-1}$ of ${\mathbf{y}}$ an \emph{$\ell$-run} if all the bits in the subsequence are the same and they differ from the bits next to the subsequence, \emph{i.e.,} $y_{i}\neq y_{i-1},y_{i+\ell-1}=\cdots=y_{i}$ and $y_{i+\ell}\neq y_{i+\ell-1}$.
Thus, let $R^{(\ell)}_{\mathbf{y}}$ be an integer counting the number of $\ell$-runs in ${\mathbf{y}}$. It follows that
\begin{align}
\label{eq:4.3}
\sum_{\ell=1}^{\lo}\ell R^{(\ell)}_{\mathbf{y}} = \lo
\end{align}
for all ${\mathbf{y}}\in\{0,1\}^{\lo}$.

\subsubsection{Approximation and Duplication Ratio}
Figuring out the the maximal number of deletion patterns $\overline{\#}\left({\mathbf{y}}\right)$ as an explicit expression using the number of $\ell$-runs $R^{(1)}_{\mathbf{y}}$,$\ldots,R^{(\lo)}_{\mathbf{y}}$ is a nontrivial task. Instead, one might turn to consider approximations of the quantity $\overline{\#}\left({\mathbf{y}}\right)$. Suppose $\digamma:={\li}/{\lo}$ is an integer. Intuitively, duplicating $\digamma$ times each bit in $\mathbf{y}$ may provide a decent estimate of $\overline{\#}\left({\mathbf{y}}\right)$, which motivates the following definition.

\begin{definition}
	\label{def:5}
	Suppose $\lo$ divides $\li$. Denoted by $\digamma:={\li}/{\lo}\in\mathbbm{N}^{+}$. We define the \emph{duplication ratio} $\gamma\left(\li,\digamma,\mathbf{y}\right)\leq 1$ of the MDM problem to be the ratio of the approximated number of deletion patterns by duplicating each bit $\digamma$ times in $\mathbf{y}$ and the maximal number of deletion patterns $\overline{\#}\left({\mathbf{y}}\right)$:
	\begin{align}
	\label{eq:4.4}
\gamma\left(\li,\digamma,\mathbf{y}\right):=\frac{\overline{\#}_{\mathrm{dup}}\left({\mathbf{y}}\right) }{\overline{\#}\left({\mathbf{y}}\right)}.
	\end{align}
where $\overline{\#}_{\mathrm{dup}}\left({\mathbf{y}}\right)$ is given by
	\begin{align}
	\label{eq:4.20}
	\overline{\#}_{\mathrm{dup}}\left({\mathbf{y}}\right):= \prod_{\ell=1}^{\lo}{\ell \digamma\choose \ell}^{R^{(\ell)}_{\mathbf{y}}}.
	\end{align}
\end{definition}

Note that $\overline{\#}_{\mathrm{dup}}\left({\mathbf{y}}\right)$ in (\ref{eq:4.20}) equals to the number of deletion patterns of the length-$\lo$ binary sequence $\mathbf{y}$ in the length-$\li$ sequence $\mathbf{x}_{\mathrm{dup}}\left(\mathbf{y}\right)$ by setting
\begin{align*}
\mathbf{x}_{\mathrm{dup}}\left(\mathbf{y}\right) = \underbrace{y_1,y_1,\ldots,y_1}_{\digamma \text{ many}},\ldots,\underbrace{y_\lo,y_\lo,\ldots,y_\lo}_{\digamma \text{ many}}.
\end{align*}
	
\subsubsection{Numerical Results}
We compute the duplication ratios with different block-lengths $\li$ and $\lo$. We exemplify part of the results in Table~\ref{table:1} and Table~\ref{table:2}. Furthermore, setting $\li=18$, we plot the following quantities:
\begin{align*}
\overline{C}_{\li}(\delep):=&\frac{1}{\li}\log\bigg(\sum_{\mathbf{y}\in\{0,1\}^{\lceil\li(1-\delep)\rceil}}\max_{\mathbf{x}\in\{0,1\}^{\li}}\#\left({\mathbf{x}},{\mathbf{y}}\right)\bigg),\\
\widetilde{C}_{\li}(\delep):=&\frac{1}{\li}\log\bigg(\sum_{\mathbf{y}\in\{0,1\}^{\lceil\li(1-\delep)\rceil}}\max_{\mathbf{x}\in\{0,1\}^{\li}}\#_{\mathrm{dup}}\left({\mathbf{x}},{\mathbf{y}}\right)\bigg).
\end{align*}

The output block-length $\lo$ in above is set to be an integer between $1$ and $\li=18$. In order to approximate the values of $\overline{C}_{\li}(\delep)$ and $\widetilde{C}_{\li}(\delep)$ when $\lo$ is not divided by $\li$, we consider the following three different approaches:

\textit{Approach 1 (Assign-to-the-last Approximation):}

First, duplicate $\digamma$ times each bit in $\mathbf{y}$ where $\digamma$ is the largest integer satisfying $\lo\digamma\leq \li$; then for the remaining $\li-\lo\digamma$ bits, assign them proportionally to the last several runs in $\mathbf{y}$. For instance, suppose $\lo=6$, $\li=15$ and $\mathbf{y}=010001$. An approximation can be obtained by first constructing a length-$12$ sequence by duplicating the bits in $\mathbf{y}$; then assigning $1$ bit to the last run and $2$ bits to the second last run.

\textit{Approach 2 (Assign-by-the-length Approximation):}

The first step is the same as Approach 1. For the remaining bits, longer runs get more bits, \textit{i.e.,} assign them to the longest run (of length $\ell$) until the length exceeds $\ell\li/\lfloor\lo\rfloor$. For example, suppose $\lo=6$, $\li=15$ and $\mathbf{y}=010001$. Then the formed new length-$\li$ sequence is $001100000000011$.

\textit{Approach 3 (Gamma Function Approximation):}
	
An alternative approximation is to substitute the binomial coefficients in Eq. (\ref{eq:4.20}) by Gamma functions, thus ensuring $\digamma=\li/\lo$ taking non-integers.
	
The three approximations of $\widetilde{C}_{\li}(\delep)$ and the ML upper bound $\overline{C}_{\li}(\delep)$ are depicted in Figure~\ref{fig:2}, together with the best known numerical upper bounds reported in~\cite{rahmati2015upper}.

\begin{figure}[H]
	\centering
	\includegraphics[scale=0.45]{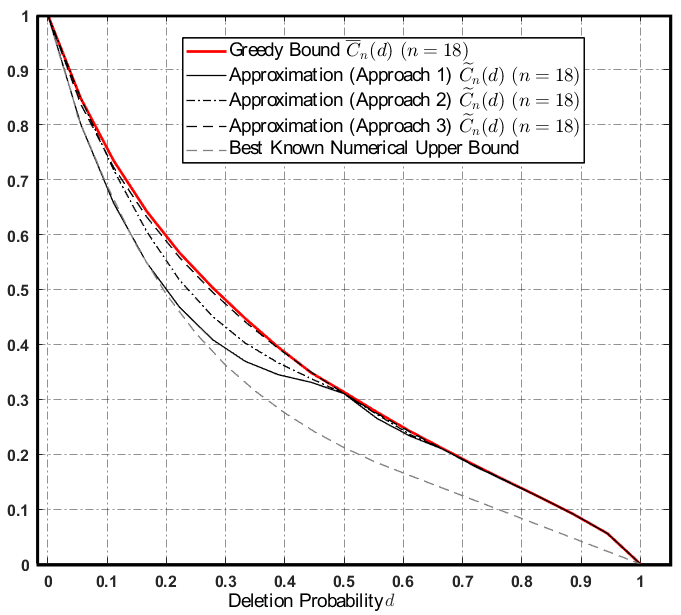}
	\caption{The ML upper bound (solid, red) $\overline{C}_{\li}(\delep)$ for BDC with block-length $\li=18$, together approximations of $\widetilde{C}_{\li}(\delep)$ obtained via the three approaches. The lower curve (dashed gray) corresponds to the best known numerical upper bounds provided in~\cite{rahmati2015upper}.}
	\label{fig:2}
	\medskip
	\hrule
\end{figure}

Some observations are summarized as hypotheses in sequel.

\begin{figure*}
	\begin{align}
	\label{eq:4.18}
	{\ell \digamma\choose \ell}= \frac{\left(\ell\digamma\right) !}{\ell !\left(\ell\digamma-\ell\right)! }
	&\leq \frac{\exp\left(-\ell\digamma+1\right) \left(\ell\digamma\right)^{\ell\digamma+1/2} }{\left(\sqrt{2\pi}\exp(-\ell)\ell^{\ell+1/2}\right)\cdot\left(\sqrt{2\pi}\exp(-\left(\ell\digamma-\ell\right))\left(\ell\digamma-\ell\right)^{\left(\ell\digamma-\ell\right)+1/2}\right)}
	\end{align}
	\hrule
\end{figure*}
\begin{hypothesis}
	\label{hypo:1}
	For any block-lengths $\li,\lo\in\mathbbm{N}^{+}$ and $\digamma={\li}/{\lo}\in\mathbbm{N}^{+}$,
	\begin{align}
	\nonumber
	\min_{\mathbf{y}\in\{0,1\}^{{\li}/{\digamma}}}\gamma\left(\li,\digamma,\mathbf{y}\right) = \frac{\digamma^{{\li}/{\digamma}}}{\overline{\#}\left({\mathbf{y}_{\mathrm{flip}}}\right)}
	\end{align}
	where ${\mathbf{y}_{\mathrm{flip}}}$ denotes a length-$\lo$ binary sequence with flipping bits, \textit{i.e.,} $R^{(1)}_{\mathbf{y}} = \lo$ and $R^{(\ell)}_{\mathbf{y}} = 0$ for all $\ell>1$.
\end{hypothesis}
Moreover, the approximations are tight, such that
\begin{hypothesis}
	\label{hypo:2}
	For any block-lengths $\li,\lo\in\mathbbm{N}^{+}$ and $\digamma={\li}/{\lo}\in\mathbbm{N}^{+}$,
	\begin{align}
	\nonumber
	\lim_{\li\rightarrow\infty}\frac{1}{\li}\log \gamma\left(\li,\digamma\right) = 0 .
	\end{align}
\end{hypothesis}

Furthermore, based on Hypothesis~\ref{hypo:1}, Hypothesis~\ref{hypo:2} is true when $\delep\rightarrow 1$.
Denote by $\gamma\left(\li,\digamma\right):=\min_{\mathbf{y}\in\{0,1\}^{{\li}/{\digamma}}}\gamma\left(\li,\digamma,\mathbf{y}\right)$ for notational convenience. We conclude the following asymptotic behavior of $\gamma\left(\li,\digamma,\mathbf{y}\right)$:
\begin{lemma}
	\begin{align*}
	\lim_{{\lo}/{\li}\rightarrow 1}\lim_{\li\rightarrow\infty}\frac{1}{\li}\log \gamma\left(\li,\digamma\right)=0.
	\end{align*}
\end{lemma}
\begin{proof}
	Noe that $\gamma\left(\li,\digamma\right) = \frac{\digamma^{{\li}/{\digamma}}}{\overline{\#}\left({\mathbf{y}_{\mathrm{flip}}}\right)}\geq\frac{\digamma^{{\li}/{\digamma}}}{{\li\choose\lo}}$. Using Stirling's approximation (see (\ref{eq:4.18})), we get
	\begin{align*}
	{\li\choose\lo}\leq \frac{e}{2\pi}\cdot \frac{2^{\li h(\frac{\lo}{\li})}}{\sqrt{\left(1-\frac{\lo}{\li}\right)\lo}}.
	\end{align*}
	
	Since $\digamma=\li/\lo$, 
	\begin{align*}
	\gamma\left(\li,\digamma\right)\geq \frac{2\pi}{e}\cdot\left(\frac{\li}{\lo}\right)^{\lo}\frac{\sqrt{\left(1-\frac{\lo}{\li}\right)\lo}}{2^{\li h(\frac{\lo}{\li})}}.
	\end{align*}
	
	Thus, taking logarithm and letting $\li\rightarrow\infty$,
	\begin{align*}
	\lim_{\li\rightarrow\infty}\frac{1}{\li}\log \gamma\left(\li,\digamma\right)\geq& \lim_{\li\rightarrow\infty}\frac{1}{\li}\left(\lo\log\frac{\li}{\lo}-\li h\left(\frac{\lo}{\li}\right)\right)\\
	=&\lim_{\li\rightarrow\infty}\left(1-\frac{\lo}{\li}\right)\log\left(\frac{1}{1-{\lo}/{\li}}\right),
	\end{align*}
	which goes to $0$ as ${\lo}/{\li}\rightarrow 1$.
\end{proof}

The remaining part of this work is based on Hypothesis~\ref{hypo:2}.

\subsection{Explicit Approximation of $\overline{C}(\delep)$}
\label{sec:4.5}

Based on Hypothesis~\ref{hypo:2}, we can further bound the capacity $C(\delep)$ of the BDC when the deletion probability $\delep\geq 1/2$.

\begin{lemma}
Suppose Hypothesis~\ref{hypo:2} is true. For any block-lengths $\li,\lo\in\mathbbm{N}^{+}$ and $\digamma={\li}/{\lo}\in\mathbbm{N}^{+}$, the following bound on $C(\delep)$ holds:
\begin{align}
\label{eq:4.17}
C(\delep) 
\leq\liminf_{\li\rightarrow\infty}\frac{1}{\li}\log\sum_{\mathbf{y}\in\{0,1\}^{\lo}}{\prod_{\ell=1}^{\lo}{\ell \digamma\choose \ell}^{R_{\ell}\left({\mathbf{y}}\right)}}-h\left(\delep\right)
\end{align}
where $\lo=\lceil\li(1-\delep)\rceil$.
\end{lemma}

\begin{proof}
We start with repeating Theorem~\ref{thm:1}:
\begin{align} 
\label{eq:4.1}
C(\delep)\leq\overline{C}(\delep)
=\frac{1}{\li}\log\Big(\sum_{\mathbf{y}\in\{0,1\}^{\lo}}\overline{\#}\left({\mathbf{y}}\right)\Big)-h\left(\delep\right)
\end{align}
provided $\lo=\lceil\li(1-\delep)\rceil$.

Suppose $\lo$ divides $\li$. Considering the definition of $\gamma\left(\li,\digamma\right)$,  the logarithmic term in above can be bounded as
\begin{align}
\nonumber
&\frac{1}{\li}\log\sum_{\mathbf{y}\in\{0,1\}^{\lo}}\overline{\#}\left({\mathbf{y}}\right)\\
\nonumber
\leq& \frac{1}{\li}\log\Bigg(\sum_{\mathbf{y}\in\{0,1\}^{\lo}}\frac{\prod_{\ell=1}^{\lo}{\ell \digamma\choose \ell}^{R^{(\ell)}_{\mathbf{y}}}}{\gamma\left(\li,\digamma\right)}\Bigg)\\
\label{eq:4.14}
=&\frac{1}{\li}\log\Bigg(\sum_{\mathbf{y}\in\{0,1\}^{\lo}}{\prod_{\ell=1}^{\lo}{\ell \digamma\choose \ell}^{R^{(\ell)}_{\mathbf{y}}}}\Bigg)-\frac{1}{\li}\log	\gamma\left(\li,\digamma\right)
\end{align}

Taking the limits $\li\rightarrow\infty$ and $\delep\rightarrow 1$, Hypothesis~\ref{hypo:2} implies (\ref{eq:4.17}).
\end{proof}

We will then show the following lemma holds:
\begin{lemma}
Suppose $\lo=\lceil\li(1-\delep)\rceil$. For all deletion probability $\delep\in [1/2,1)$,
\begin{align}
\nonumber
&\frac{1}{\li}\log\sum_{\mathbf{y}\in\{0,1\}^{\lo}}{\prod_{\ell=1}^{\lo}{\ell \digamma\choose \ell}^{R_{\ell}\left({\mathbf{y}}\right)}}\\
\label{eq:4.22}
=&h(\delep)+1-\delep+\frac{1}{\li}\log\Big(\mathbbm{E}\big[\exp\left({-\mu_{\delep}\left({\mathbf{Y}}\right)}\right)\big]\Big)
\end{align}
where 
\begin{align*}
\mu_{\delep}\left({\mathbf{y}}\right):=&\frac{1}{2}\sum_{\ell=1}^{\lo}R^{(\ell)}_{\mathbf{y}}\ln\left(\left(\frac{2\pi}{e}\right)^2\delep\ell\right),
\end{align*}
and the expectation is over a random length-$\lo$ sequence $\mathbf{Y}$ wherein each bit is a random variable with distribution $\mathrm{Bernoulli}\left({1}/{2}\right)$.
\end{lemma}
\begin{proof}
Applying Stirling's approximation (inequalities) to the binomial coefficients, (\ref{eq:4.17}) follows. Therefore
\begin{align}
\nonumber
{\ell \digamma\choose \ell}&\leq\frac{e\left(\ell\digamma\right)^{\ell\digamma+1/2}}{2\pi\ell^{\ell+1/2}\left(\ell\digamma-\ell\right)^{\left(\ell\digamma-\ell\right)+1/2}}\\
\label{eq:4.13}
&=\frac{e}{2\pi}\cdot\frac{\left(\ell\digamma\right)^{\ell\digamma}}{\ell^{\ell}\left(\ell\digamma-\ell\right)^{\left(\ell\digamma-\ell\right)}}\cdot \sqrt{\frac{\digamma}{\ell\left(\digamma-1\right)}}.
\end{align}

Since $\digamma=\frac{1}{1-\delep}$, we have $\frac{\digamma}{\digamma-1}=\frac{1}{\delep}$. Thus,
\begin{align}
\label{eq:4.11}
\frac{\left(\ell\digamma\right)^{\ell\digamma}}{\ell^{\ell}\left(\ell\digamma-\ell\right)^{\left(\ell\digamma-\ell\right)}}&=2^{\ell\digamma h\left({\left(\digamma-1\right)}/{\digamma}\right)} = 2^{\ell\digamma h(\delep)},\\
\label{eq:4.12}
\sqrt{\frac{\digamma}{\ell\left(\digamma-1\right)}}& = \sqrt{\frac{1}{\delep\ell}}.
\end{align}

Putting (\ref{eq:4.11}) and (\ref{eq:4.12}) into (\ref{eq:4.13}), 
\begin{align*}
{\ell \digamma\choose \ell}\leq \frac{e}{2\pi}\cdot \frac{2^{\ell\digamma h(\delep)}}{\sqrt{\delep\ell}}.
\end{align*}

Therefore,
\begin{align*}
&\ln\left(\prod_{\ell=1}^{\lo}{\ell \digamma\choose \ell}^{R^{(\ell)}_{\mathbf{y}}}\right) \\
=& \sum_{\ell=1}^{\lo}R^{(\ell)}_{\mathbf{y}}\ln{\ell \digamma\choose \ell}\\
\leq&\sum_{\ell=1}^{\lo}\left({\ln 2}\right)\ell R^{(\ell)}_{\mathbf{y}} \digamma h(\delep)\\
&+\sum_{\ell=1}^{\lo}R^{(\ell)}_{\mathbf{y}}\ln\frac{e}{2\pi} -\frac{1}{2}\sum_{\ell=1}^{\lo}R^{(\ell)}_{\mathbf{y}}\ln\left(\delep\ell\right).
 \end{align*}

According to (\ref{eq:4.3}), $\sum_{\ell=1}^{\lo}\ell R^{(\ell)}_{\mathbf{y}} = \lo$, implying that
\begin{align*}
\sum_{\ell=1}^{\lo}\ell R^{(\ell)}_{\mathbf{y}} \digamma h(\delep) =\lo \digamma h(\delep) = \li h(\delep).
\end{align*}

Continuing from above,
\begin{align}
\nonumber
&\ln\left(\prod_{\ell=1}^{\lo}{\ell \digamma\choose \ell}^{R^{(\ell)}_{\mathbf{y}}}\right) \\
\leq&\left({\ln 2}\right)\li h(\delep)+ \sum_{\ell=1}^{\lo}R^{(\ell)}_{\mathbf{y}}\ln\frac{e}{2\pi} -\frac{1}{2}\sum_{\ell=1}^{\lo}R^{(\ell)}_{\mathbf{y}}\ln\left(\delep\ell\right)\\
\label{eq:4.15}
=&\left({\ln 2}\right)\li h(\delep) -\frac{1}{2}\sum_{\ell=1}^{\lo}R^{(\ell)}_{\mathbf{y}}\ln\left(\left(\frac{2\pi}{e}\right)^2\delep\ell\right).
\end{align}

Recall that
\begin{align*}
\mu_{\delep}\left({\mathbf{y}}\right):=\frac{1}{2}\sum_{\ell=1}^{\lo}R^{(\ell)}_{\mathbf{y}}\ln\left(\left(\frac{2\pi}{e}\right)^2\delep\ell\right).
\end{align*}
Thus, summing over all $\mathbf{y}\in\{0,1\}^{\lo}$, (\ref{eq:4.15}) yields that
\begin{align}
\nonumber
&\frac{1}{\li}\log\sum_{\mathbf{y}\in\{0,1\}^{\lo}}{\prod_{\ell=1}^{\lo}{\ell \digamma\choose \ell}^{R_{\ell}\left({\mathbf{y}}\right)}}\\
\leq&\frac{1}{\li}\log \sum_{\mathbf{y}\in\{0,1\}^{\lo}}\exp\big(\left({\ln 2}\right){\li h(\delep)-\mu_{\delep}\left({\mathbf{y}}\right)}\big) \\
\nonumber
=&h(\delep)+\frac{1}{\li}\log\sum_{\mathbf{y}\in\{0,1\}^{\lo}}\exp\big({-\mu_{\delep}\left({\mathbf{y}}\right)}\big).
\end{align}

The summation $\sum_{\mathbf{y}\in\{0,1\}^{\lo}}\exp\left({-\mu_{\delep}\left({\mathbf{y}}\right)}\right)$ can be regarded as $2^\lo$ times the expectation of $\exp\left({-\mu_{\delep}\left({\mathbf{Y}}\right)}\right)$ given that each bit in $\mathbf{Y}$ is selected $\mathrm{Bernoulli}\left({1}/{2}\right)$. Therefore,
\begin{align}
\nonumber
&\frac{1}{\li}\log\sum_{\mathbf{y}\in\{0,1\}^{\lo}}{\prod_{\ell=1}^{\lo}{\ell \digamma\choose \ell}^{R_{\ell}\left({\mathbf{y}}\right)}}
\\
\nonumber
\leq&h(\delep)+\frac{1}{\li}\log\Big(2^\lo\mathbbm{E}\big[\exp\left({-\mu_{\delep}\left({\mathbf{Y}}\right)}\right)\big]\Big)\\
\label{eq:4.21}
=&h(\delep)+1-\delep+\frac{1}{\li}\log\Big(\mathbbm{E}\big[\exp\left({-\mu_{\delep}\left({\mathbf{Y}}\right)}\right)\big]\Big).
\end{align}
\end{proof}


Taking the expectation outside, we derive the following approximation of the last term in (\ref{eq:4.21}):
\begin{align}
\label{eq:4.24}
\frac{1}{\li}\mathbbm{E}\Big[\log\big(\exp\left({-\mu_{\delep}\left({\mathbf{Y}}\right)}\right)\big)\Big]=-\frac{1}{\li}\mathbbm{E}\big[{\left(\log e\right)\mu_{\delep}\left({\mathbf{Y}}\right)}\big].
\end{align}
For a $\mathrm{Bernoulli}\left({1}/{2}\right)$ process, the distribution of the number of $\ell$-runs $R^{(1)}_{\mathbf{y}}$,$\ldots,R^{(\lo)}_{\mathbf{y}}$ is proportional to the ``run-length distribution'' defined in~\cite{kanoria2013optimal}, which is $\{1/2^{\ell}\}_{\ell=1}^{\infty}$. Hence, considering (\ref{eq:4.3}), the expectation of the number of runs is ${\lo}/{2^{\ell+1}}$, \textit{i.e.,}
\begin{align}
\nonumber
\mathbbm{E}\big[R^{(\ell)}_{\mathbf{Y}}\big] = \frac{\lo}{2^{\ell+1}}.
\end{align}
It follows that
\begin{align}
\nonumber
\mathbbm{E}\big[{\mu_{\delep}\left({\mathbf{Y}}\right)}\big]=&\mathbbm{E}\Bigg[\frac{1}{2}\sum_{\ell=1}^{\lo}R^{(\ell)}_{\mathbf{y}}\ln\left(\left(\frac{2\pi}{e}\right)^2\delep\ell\right)\Bigg]\\
\nonumber
=&\frac{1}{2}\sum_{\ell=1}^{\lo}\mathbbm{E}\big[R^{(\ell)}_{\mathbf{Y}}\big]\ln\left(\left(\frac{2\pi}{e}\right)^2\delep\ell\right)\\
\label{eq:4.23}
=&\frac{\lo}{2}\sum_{\ell=1}^{\lo}\frac{1}{2^{\ell+1}}\ln\left(\left(\frac{2\pi}{e}\right)^2\delep\ell\right).
\end{align}

Combining (\ref{eq:4.17}), (\ref{eq:4.22}), the approximation (\ref{eq:4.23}) and the identity (\ref{eq:4.24}) above, the approximation $\widetilde{C}(\delep)$ defined below holds:
\begin{align}
\nonumber
\widetilde{C}(\delep):=&1-\delep-\lim_{\li\rightarrow\infty}\frac{1}{\li}\mathbbm{E}\big[{\left(\log e\right)\mu_{\delep}\left({\mathbf{Y}}\right)}\big]\\
\label{eq:4.25}
=&1-\delep-\frac{1}{2}\psi\left(1-\delep\right)
\end{align}
where
\begin{align*}
\psi:=&\sum_{\ell=1}^{\infty}\frac{1}{2^{\ell+1}}\log\left(\left(\frac{2\pi}{e}\right)^2\delep\ell\right)\\
=&\frac{1}{2}\log\delep+\sum_{\ell=1}^{\infty}\frac{\log\left(\left({2\pi}/{e}\right)\sqrt{\ell}\right)}{2^{\ell}} \lesssim\frac{1}{2}\log\delep+1.09179.
\end{align*}

A figure depicting the explicit approximation $\widetilde{C}(\delep)$ for $\delep\geq 1/2$ is provided below.

\begin{figure}[H]
	\centering
	\includegraphics[scale=0.45]{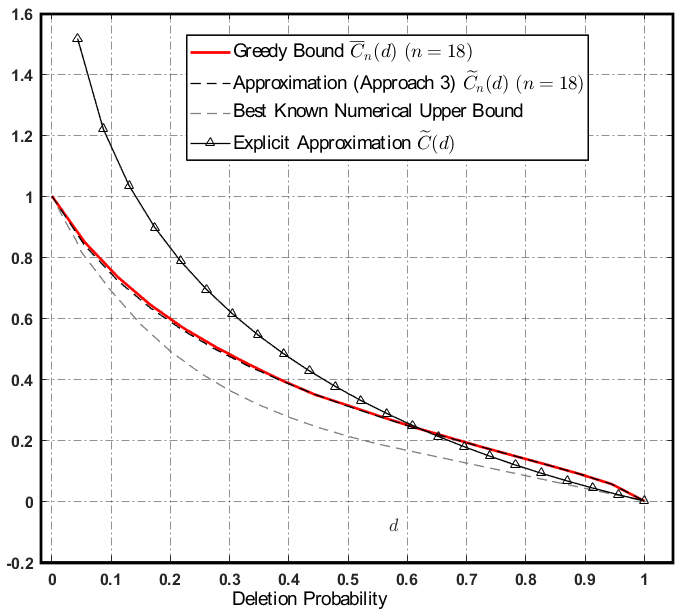}
	\caption{The ML upper bound (solid, red) $\overline{C}_{\li}(\delep)$ for BDC with block-length $\li=18$, the approximation of $\widetilde{C}_{\li}(\delep)$ obtained via the gamma function, and the \textit{explicit approximation} $\widetilde{C}(\delep)$ derived in (\ref{eq:4.25}). The lower curve (dashed gray) corresponds to the best known numerical upper bounds provided in~\cite{rahmati2015upper}.}
	\label{fig:3}
	\medskip
	\hrule
\end{figure}

\section{Conclusion}

We derive a general ML upper bound (See~Theorem~\ref{thm:0}) for \textit{information stable channels}. The corresponding bounds are shown to be tight for simple channels, \textit{e.g.,} the BEC and the BSC. Furthermore, we demonstrate the usage of the bound on the BDC, whose capacity remains unknown. The corresponding upper bound for the BDC derived from the general bound coincides with an intriguing combinatorial problem (defined as the MDM problem in Definition~\ref{def:4}). Approximations for the derived upper bound are provided via three different approaches. Furthermore, analyzing and approximating the limiting behavior of the derived upper bound gives an explicit bound reported in (\ref{eq:4.25}) (and shown in Figure~\ref{fig:3}), validating that the general bound is capable of providing nontrivial results for sophisticated channels with memory. The next step is to validate the upper bounds on a varriety types of channels and formalize a more general framework based on the main result (Theorem~\ref{thm:0}) stated in the paper.

\newpage
\addcontentsline{toc}{section}{Bibliography}
{\bibliographystyle{IEEEtran}
	\bibliography{ref}}

\end{document}